\documentclass[11pt]{article}

\usepackage{amsmath}
\usepackage{subcaption}
\usepackage{xspace}
\usepackage{bbm}
\usepackage{enumerate}
\usepackage{wrapfig}
\usepackage{framed}
\usepackage{algorithm,algorithmicx}
\usepackage[noend]{algpseudocode}
\usepackage{thmtools, thm-restate}
\usepackage[framemethod=TikZ]{mdframed}

\usepackage[colorlinks,citecolor=blue,linkcolor=red]{hyperref}
\usepackage{cleveref}

\newcommand{\mcA}{\mathcal{A}}

\newcommand{\mcC}{\mathcal{C}}
\newcommand{\mcD}{\mathcal{D}}

\newcommand{\mcK}{\mathcal{K}}

\newcommand{\mcP}{\mathcal{P}}

\newcommand{\E}{\mathop{{}\mathbb{E}}}%Expectation
\newcommand{\con}{\text{con}}%congestion
\newcommand{\dil}{\text{dil}}%dilation
\DeclareMathOperator*{\argmax}{arg\,max}%argmax
\DeclareMathOperator*{\argmin}{arg\,min}%argmin
\newcommand{\bc}[1]{&\text{$\left(\text{By #1} \right)$}}%\bc
\newcommand{\poly}{\mathrm{poly}} %poly
\newcommand{\OPT}{\mathrm{OPT}} %OPT
\newcommand{\gum}{\textsc{Token Network}\xspace}
\newcommand{\gump}{\textsc{Token Computation}\xspace}
\newcommand{\solvegump}{\textsc{SolveTC}\xspace}
\newcommand{\wad}{token\xspace}
\newcommand{\wads}{tokens\xspace}
 
\DeclareMathOperator*{\eps}{\epsilon} %eps

\usepackage{amsfonts}
\usepackage{mathtools}
\usepackage{amsthm}
\usepackage{fullpage}
\usepackage{authblk}
\newtheorem{theorem}{Theorem}

\newtheorem{lemma}{Lemma}
\newtheorem{definition}{Definition}

\begin{document}
	% Title portion. Note the short title for running heads
	\title{Computation-Aware Data Aggregation}
	\newcommand*\samethanks[1][\value{footnote}]{\footnotemark[#1]}
	%\author{Submission 143}
	
	\author[ ]{Bernhard Haeupler\thanks{Supported in part by NSF grants CCF-1527110, CCF-1618280, CCF-1814603, 
			CCF-1910588, NSF CAREER award CCF-1750808 and a Sloan Research Fellowship.}, D Ellis Hershkowitz\samethanks, Anson Kahng\thanks{Supported in part by NSF grants IIS-1350598, IIS-1714140, CCF-1525932, and CCF-1733556; by ONR grants N00014-16-1-3075 and N00014-17-1-2428; and by a Sloan Research Fellowship and a Guggenheim Fellowship.}, Ariel D. Procaccia\samethanks\\ 
		Computer Science Department\\
		Carnegie Mellon University\\
		\texttt{\{haeupler,dhershko,akahng,arielpro\}@cs.cmu.edu}}
	\date{}
	\maketitle

\begin{abstract}
	Data aggregation is a fundamental primitive in distributed computing wherein a network computes a function of every nodes' input. However, while compute time is non-negligible in modern systems, standard models of distributed computing do not take compute time into account. Rather, most distributed models of computation only explicitly consider communication time.
	
	In this paper, we introduce a model of distributed computation that considers \emph{both} computation and communication so as to give a theoretical treatment of data aggregation. We study both the structure of and how to compute the fastest data aggregation schedule in this model. As our first result, we give a polynomial-time algorithm that computes the optimal schedule when the input network is a complete graph. Moreover, since one may want to aggregate data over a pre-existing network, we also study data aggregation scheduling on arbitrary graphs. We demonstrate that this problem on arbitrary graphs is hard to approximate within a multiplicative $1.5$ factor. Finally, we give an $O(\log n \cdot \log \frac{\mathrm{OPT}}{t_m})$-approximation algorithm for this problem on arbitrary graphs, where $n$ is the number of nodes and $\mathrm{OPT}$ is the length of the optimal schedule.
\end{abstract}

\section{Introduction}\label{sec:intro}
Distributed systems drive much of the modern computing revolution. However, these systems are only as powerful as the abstractions which enable programmers to make use of them. A key such abstraction is data aggregation, wherein a network computes a function of every node's input.
For example, if every node stored an integer value, a programmer could run data aggregation to compute the sum or the largest value of every node in the network. 
Indeed, the well-studied and widely-used AllReduce abstraction~\cite{rabenseifner2004optimization, grama2003introduction} consists of a data aggregation step followed by a broadcast step. 
%Indeed, in AllReduce, a programmer specifies some function they would like to apply to every private input of a node in a network; after AllReduce is run, a node learns the result of this function applied to every nodes' input, and the value is then broadcast to all nodes throughout the network. 
%Because this broadcast step only involves communication and because once one node has learned this value, it is well-understood how to spread said value in the network \cite{ravi1994rapid, kortsarz1995approximation, iglesias2015rumors}, we focus just on the ``compute'' step, or data aggregation. 

The utility of modern systems is their ability to perform massive computations and so, applications of data aggregation often consist of a function which is computationally-intensive to compute. A rigorous theoretical study of data aggregation, then, must take the cost of computation into account. At the same time, one cannot omit the cost of communication, as many applications of data aggregation operate on large datasets which take time to transmit over a network.

However, to our knowledge, all existing models of distributed computation---e.g., the CONGEST \cite{peleg2000distributed}, SINR \cite{andrews2009maximizing}, (noisy) radio network \cite{kushilevitz1998computation,chlamtac1985broadcasting,censor2017broadcasting,censor2018erasure}, congested clique \cite{drucker2014power}, dual graph \cite{censor2014structuring}, store-and-forward \cite{leighton1994packet, rothvoss2013simpler}, LOCAL \cite{linial1992locality}, and telephone broadcast models \cite{ravi1994rapid, kortsarz1995approximation, iglesias2015rumors}---all only consider the cost of communication. Relatedly, while there has been significant applied research on communication-efficient data aggregation algorithms, there has been relatively little work that explicitly considers the cost of computation, and even less work that considers how to design a network to efficiently perform data aggregation \cite{oden2014energy, klenk2015analyzing,patarasuk2007bandwidth, patarasuk2009bandwidth,jain2012collectives}. In this way, there do not seem to exist theoretical results for efficient data aggregation scheduling algorithms that consider both the cost of communication and computation. 

Thus, we aim to provide answers to two theoretical questions in settings where both computation and communication are non-negligible: 
\begin{enumerate}
	\item \textit{How should one structure a network to efficiently perform data aggregation?}
	\item \textit {How can one coordinate a fixed network to efficiently perform data aggregation?}
\end{enumerate}

%\noindent this paper, we present a simple distributed model of computation which takes both communication and computation into account. First, we show how to optimally schedule AllReduce in this model on a network in which every node can communicate with every other node. Next, we demonstrate that no polynomial-time algorithm can compute an AllReduce schedule whose length is within a multiplicative $1.5$ factor of the shortest schedule. Finally, we give an $O(\log n \cdot \log \frac{\mathrm{OPT}}{t_m})$-approximation algorithm for computing the shortest AllReduce schedule on arbitrary graphs, where $n$ is the number of nodes and $\mathrm{OPT}$ is the length of the optimal schedule.

%In particular, we study this from the perspective of the designer of an organization as opposed to individual members of the organization; our results are centralized, not decentralized.
% How does the structure of an organization affect its overall performance?

\subsection{Our Model and Problem}
%We now give a description of our model and problem and discuss our modeling choices below. %We defer more formal definitions to \Cref{appsec:defs}.
\paragraph{The \gum Model.}  So as to give formal answers to these questions we introduce the following simple distributed model, the \gum Model. A \gum is given by an undirected graph $G = (V,E)$, $|V|=n$, with parameters $t_c, t_m \in \mathbb{N}$ which describe the time it takes nodes to do computation and communication, respectively.\footnote{We assume $t_c, t_m = \poly(n)$ throughout this paper.} %$E$ represents the communication network of the graph: an edge $(i,j)$ means that node $i$ and node $j$ can send \wad to one another. 

Time proceeds in synchronous rounds during which nodes can compute on or communicate atomic tokens. Specifically, in any given round a node is busy or not busy. If a node is not busy and has at least one \wad it can \emph{communicate}: any node that does so is busy for the next $t_m$ rounds, at the end of which it passes one of its \wads to a neighbor in $G$. If a node is not busy and has two or more \wads, it can \emph{compute}: any node that does so is busy for the next $t_c$ rounds, at the end of which it combines (a.k.a.\ aggregates) two of its \wads into a single new \wad.\footnote{Throughout this paper, we assume for ease of exposition that the smaller of $t_c$ and $t_m$ evenly divides the larger of $t_c$ and $t_m$, or equivalently that either $t_c$ or $t_m$ is $1$.} At a high level, this means that communication takes $t_m$ rounds and computation takes $t_c$ rounds.

\paragraph{The \gump Problem.} We use our \gum model to give a formal treatment of data aggregation scheduling. In particular, we study the \gump problem. Given an input \gum, an algorithm for the \gump problem must output a schedule $S$ which directs each node when to compute and when and with whom to communicate. A schedule is valid if after the schedule is run on the input \gum where every node begins with a single \wad, there is one remaining \wad in the entire network; i.e., there is one node that has aggregated all the information in the network. We use $|S|$ to notate the length of $S$\,---\,i.e., the number of rounds $S$ takes\,---\,and measure the quality of an algorithm by the length of the schedule that it outputs. For completeness, we give a more technical and formal definition in \Cref{appsec:defs}.

\paragraph{Discussion of Modeling Choices.} Our \gum model and the \gump problem are designed to formally capture the challenges of scheduling distributed computations where both computation and communication are at play. In particular, combining \wads can be understood as applying some commutative, associative function to the private input of all nodes in a network. For instance, summing up private inputs, taking a minimum of private inputs, or computing the intersection of input sets can all be cast as instances of the \gump problem.  % we assume that that the computed function is commutative and associative as this was the case for that experiment. 
%Additionally, unlike most models of distributed computing, we assume that both the computation and the communication by nodes take time. 
We assume that the computation time is the same for every operation and that the output of a computation is the same size as each of the inputs as a simplifying assumption. We allow nodes to receive information from multiple neighbors as this sort of communication is possible in practice.

Lastly, our model should be seen as a so-called ``broadcast'' model \cite{kushilevitz1998computation} of communication. In particular, it is easy to see that our assumption that a node can send its token to only a single neighbor rather than multiple copies of its token to multiple neighbors is without loss of generality: One can easily modify a schedule in which nodes send multiple copies to one of equal length in which a node only ever sends one token per round. An interesting followup question could be to consider our problem in a non-broadcast setting.

\subsection{Our Results}
We now give a high-level description of our technical results.

\paragraph{Optimal Algorithm on Complete Graphs (\Cref{sec:optcomplete}).} We begin by considering how to construct the optimal data aggregation schedule in the \gum model for complete graphs for given values of $t_c$ and $t_m$. The principal challenge in constructing such a schedule is formalizing how to optimally pipeline computation and communication and showing that any valid schedule needs at least as many rounds as one's constructed schedule. 
%Moreover, showing that a schedule is optimal requires demonstrating that any valid schedule needs at least as many rounds as that schedule; however, it is a priori unclear how to show lower bounds for a schedule. 
We overcome this challenge by showing how to modify a given optimal schedule into an efficiently computable one in a way that preserves its pipelining structure. Specifically, we show that one can always modify a valid optimal schedule into another valid optimal schedule with a well-behaved recursive form. We show that this well-behaved schedule can be computed in polynomial time. Stronger yet, we show that the edges over which communication takes place in this schedule induce a tree. It is important to emphasize that this result has implications beyond producing the optimal schedule for a complete graph; it shows one optimal way to \emph{construct} a network for data aggregation (if one had the freedom to include any edge), thereby suggesting an answer to the first of our two research questions.

\paragraph{Hardness and Approximation on Arbitrary Graphs (\Cref{sec:arbgraphs}).}
We next consider the hardness of producing good schedules efficiently for arbitrary graphs and given values of $t_c$ and $t_m$. We first show that no polynomial-time algorithm can produce a schedule of length within a multiplicative $1.5$ factor of the optimal schedule unless $\text{P} = \text{NP}$. This result implies that one can only \emph{coordinate} data aggregation over a pre-existing network so well.

Given that an approximation algorithm is the best one can hope for, we next give an algorithm which in polynomial time produces an approximately-optimal \gump schedule. Our algorithm is based on the simple observation that after $O(\log n)$ repetitions of pairing off nodes with \wads, having one node in each pair route a \wad to the other node in the pair, and then having every node compute, there will be a single \wad in the network. The difficulty in this approach lies in showing that one can route pairs of \wads in a way that is competitive with the length of the optimal schedule. We show that by considering the paths in $G$ traced out by \wads sent by the optimal schedule, we can get a concrete hold on the optimal schedule. Specifically, we show that a polynomial-time algorithm based on our observation produces a valid schedule of length $O(\OPT \cdot \log n \cdot \log \frac{\OPT}{t_m})$ with high probability,\footnote{Meaning at least $1 - 1/\text{poly}(n)$ henceforth.} where $\OPT$ is the length of the optimal schedule. Using an easy bound on $\OPT$, this can be roughly interpreted as an $O(\log ^2 n)$-approximation algorithm. This result shows that data aggregation over a pre-existing network can be \emph{coordinated} fairly well.\\

%We believe that these results motivate interesting empirical directions. Previous sociology experiments have been performed in networks with decentralized coordination. Running an experiment similar to that of Bavelas and Leavitt using our optimal networks with a centralized authority would measure to what degree optimal network structure and a centralized authority expedite the time groups take to solve tasks.

Furthermore, it is not hard to see that when $t_c = 0$ and $t_m > 0$, or when $t_c > 0$ and $t_m = 0$, our problem is trivially solvable in polynomial time. However, we show hardness for the case where $t_c, t_m > 0$, which gives a formal sense in which computation and communication cannot be considered in isolation, as assumed in previous models of distributed computation.

\subsection{Terminology}

For the remainder of this paper we use the following terminology. A \wad $a$ \emph{contains} \wad $a'$ if $a = a'$ or $a$ was created by combining two \wads, one of which contains $a'$. For shorthand we write $a' \in a$ to mean that $a$ contains $a'$. A \emph{singleton} \wad is a \wad that only contains itself; i.e., it is a \wad with which a node started. We let $a_v$ be the singleton \wad with which vertex $v$ starts and refer to $a_v$ as $v$'s singleton \wad. The \emph{size} of a \wad is the number of singleton \wads it contains. Finally, let $a_f$ be the last \wad of a valid schedule $S$; the \emph{terminus} of $S$ is the node at which $a_f$ is formed by a computation.

\section{Related Work}\label{sec:relatedWork}
%- Related work: expand on related models in distributed computing, and on other related papers in sociology. 
%Though we have reviewed the primary related work in sociology in \Cref{sec:intro}, problems related to ours have been studied in computer science and economics. We focus here on the work in computer science, as it is more closely related. 

%\subsection{Related Theoretical Work}

% Data aggregation 
%Data aggregation~\cite{cornejo2012aggregation}, job scheduling~\cite{awerbuch1992competitive}, computing DAG-represented functions~\cite{vyavahare2016optimal}.
Cornejo et al.~\cite{cornejo2012aggregation} study a form of data aggregation in networks that change over time, where the goal is to collect tokens at as few nodes as possible after a certain time. However, they do not consider computation time and they measure the quality of their solutions with respect to the optimal offline algorithm. Awerbuch et al.~\cite{awerbuch1992competitive} consider computation and communication in a setting where jobs arrive online at nodes, and nodes can decide whether or not to complete the job or pass the job to a neighbor. However, they study the problem of job scheduling, not data aggregation, and, again, they approach the problem from the perspective of competitive analysis with respect to the optimal offline algorithm. 

Another line of theoretical work related to our own is a line of work in centralized algorithms for scheduling information dissemination \cite{ravi1994rapid, kortsarz1995approximation, iglesias2015rumors}. In this problem, an algorithm is given a graph and a model of distributed communication, and must output a schedule that instructs nodes how to communicate in order to spread some information. For instance, in one setting an algorithm must produce a schedule which, when run, broadcasts a message from one node to all other nodes in the graph. The fact that these problems consider spreading information is complementary to the way in which we consider consolidating it. However, we note that computation plays no role in these problems, in contrast to our \gump problem. 

Of these prior models of communication, the model which is most similar to our own is the telephone broadcast model. In this model in each round a node can ``call'' another node to transmit information or receive a call from a single neighbor. Previous results have given a hardness of approximation of $3$ \cite{elkin2005combinatorial} for broadcasting in this model and logarithmic as well as sublogarithmic approximation algorithms for broadcasting \cite{elkin2006sublogarithmic}. The two notable differences between this model and our own are (1) in our model nodes can receive information from multiple neighbors in a single round\footnote{See above for the justification of this assumption.} and (2) again, in our model computation takes a non-negligible amount of time. Note, then, that even in the special case when $t_c = 0$, our model does not generalize the telephone broadcast model; as such we do not immediately inherit prior hardness results from the telephone broadcast problem. Furthermore, (1) and especially (2) preclude the possibility of an easy reduction from our problem to the telephone broadcast problem. 

There is also a great deal of related applied work; additional details are in \Cref{app:relatedWork}.

%In the setting of radio networks, researchers have examined the Convergecast scheduling problem, which resembles the aggregation aspect of the \gump problem \cite{erzin2016convergecast, kortsarz2015radio}. However, this work does not consider the cost of computation. Additionally, researchers have studied algorithms for the general problem of multi-level aggregation, which shares many of the same applications as our work \cite{bienkowski2015online, buchbinder2017depth, papadimitriou1996computational}. Again, these works do not consider the cost of computation, and, moreover, they only examine the online setting in which requests for information aggregation arrive over time.

\section{Optimal Algorithm for Complete Graphs}
\label{sec:optcomplete}
In this section we provide an optimal polynomial-time algorithm for the \gump problem on a complete graph. The schedule output by our algorithm ultimately only uses the edges of a particular tree, and so, although we reason about our algorithm in a fully connected graph, in reality our algorithm works equally well on said tree. This result, then, informs the design of an optimal network. 

\subsection{Binary Trees (Warmup)} 

We build intuition by considering a natural solution to \gump on the complete graph: na\"ive aggregation on a rooted binary tree. In this schedule, nodes do computations and communications in lock-step. In particular, consider the schedule $S$ which alternates the following two operations until only a single node with \wads remains on a fixed binary tree: (1) every non-root node that has a \wad sends its \wad to its parent in the binary tree; (2) every $v$ with at least two \wads performs one computation. Once only one node has any \wads, that node performs computation until only a single \wad remains. After $\log n$ iterations of this schedule, the root of the binary tree is the only node with any \wads, and thereafter only performs computation for the remainder of $S$.  However, $S$ does not efficiently pipeline communication and computation: after each iteration of (1) and (2), the root of the tree gains an extra \wad. Therefore, after $\log n$ repetitions of this schedule, the root has $\log n$ tokens. In total, then, this schedule aggregates all \wads after essentially $\log n(t_c + t_m) + \log n \cdot t_c$ rounds. See \Cref{fig:binTreeEG}.

\begin{figure*}[t]
	\centering
	\begin{subfigure}[t]{0.23\textwidth}
		\centering
		\includegraphics[scale=0.07,page=1]{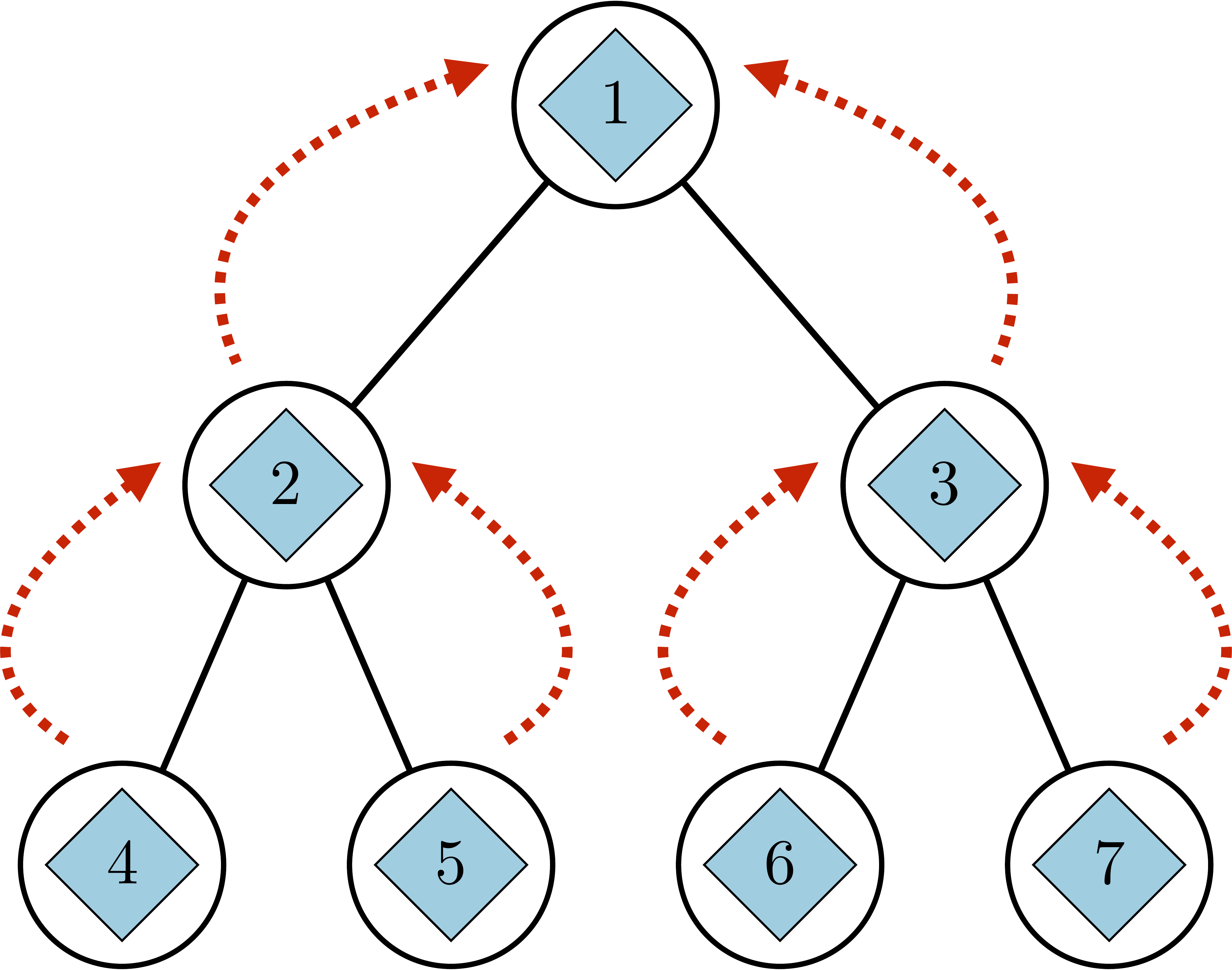}
		\caption{Round $1$}
	\end{subfigure}%
	~
	\begin{subfigure}[t]{0.23\textwidth}
		\centering
		\includegraphics[scale=0.07,page=2]{figs/binaryTreeEG.pdf}
		\caption{Round $2$}
	\end{subfigure}%
	~ 
	\begin{subfigure}[t]{0.23\textwidth}
		\centering
		\includegraphics[scale=0.07,page=3]{figs/binaryTreeEG.pdf}
		\caption{Round $3$}
	\end{subfigure}%
	~ 
	\begin{subfigure}[t]{0.23\textwidth}
		\centering
		\includegraphics[scale=0.07,page=4]{figs/binaryTreeEG.pdf}
		\caption{Round $4$}
	\end{subfigure}%
	\caption{The na\"ive aggregation schedule on a binary tree for $t_c = t_m = 1$ and $n = 7$ after $4$ rounds. \wads are represented by blue diamonds; a red arrow from node $u$ to node $v$ means that $u$ sends to $v$; and a double-ended blue arrow between two \wads $a$ and $b$ means that $a$ and $b$ are combined at the node. Notice that the root gains an extra \wad every $2$ rounds.}
	\label{fig:binTreeEG}
\end{figure*}

\begin{figure}[t]
	\centering
	\begin{subfigure}[t]{0.23\textwidth}
		\centering
		\includegraphics[scale=0.07,page=1]{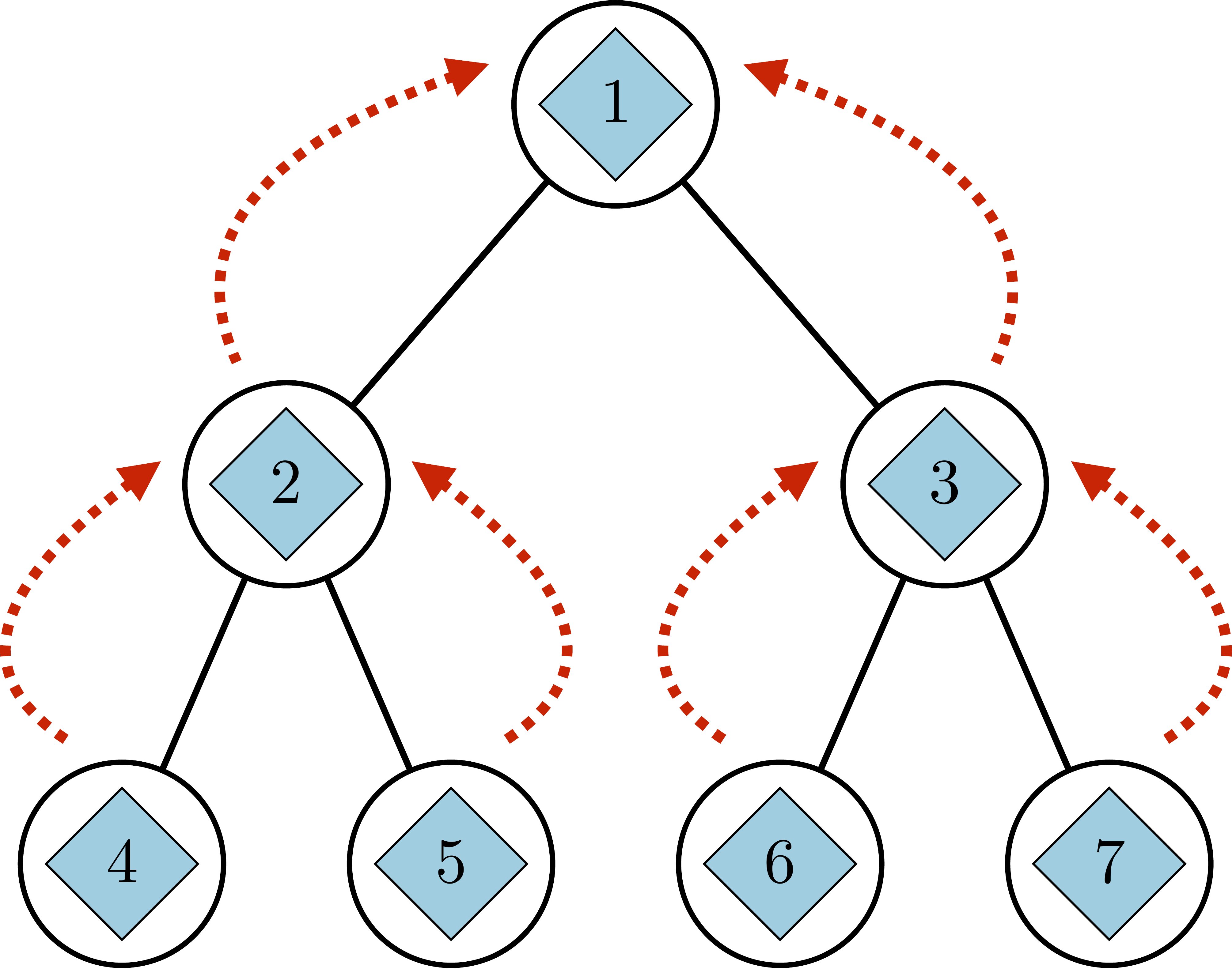}
		\caption{Round $1$}
	\end{subfigure}%
	~
	\begin{subfigure}[t]{0.23\textwidth}
		\centering
		\includegraphics[scale=0.07,page=2]{figs/binaryTreeEGPipe.pdf}
		\caption{Round $2$}
	\end{subfigure}%
	~ 
	\begin{subfigure}[t]{0.23\textwidth}
		\centering
		\includegraphics[scale=0.07,page=3]{figs/binaryTreeEGPipe.pdf}
		\caption{Round $3$}
	\end{subfigure}%
	~ 
	\begin{subfigure}[t]{0.23\textwidth}
		\centering
		\includegraphics[scale=0.07,page=4]{figs/binaryTreeEGPipe.pdf}
		\caption{Round $4$}
	\end{subfigure}%
	\caption{The aggregation schedule on a binary tree for $t_c = t_m = 1$ and $n = 7$ after $4$ rounds where the root pipelines its computations. Again, \wads are represented by blue diamonds; a red arrow from node $u$ to node $v$ means that $u$ sends to $v$; and a double-ended blue arrow between two \wads $a$ and $b$ means that $a$ and $b$ are combined at the node. Notice that the root will never have more than 3 \wads when this schedule is run.}
	\label{fig:binTreeEGPipe}
\end{figure}

For certain values of $t_c$ and $t_m$, we can speed up na\"ive aggregation on the binary tree by pipelining the computations of the root with the communications of other nodes in the network. In particular, consider the schedule $S'$ for a fixed binary tree for the case when $t_c = t_m$ in which every non-root node behaves exactly as it does in $S$ but the root always computes. Since the root always computes in $S'$, even as other nodes are sending, it does not build up a surplus of tokens as in $S$. Thus, this schedule aggregates all \wads after essentially $\log n (t_c + t_m)$ rounds when $t_c = t_m$, as shown in \Cref{fig:binTreeEGPipe}. 

However, as we will prove, binary trees are not optimal even when they pipeline computation at the root and  $t_c = t_m$. In the remainder of this section, we generalize this pipelining intuition to arbitrary values of $t_c$ and $t_m$ and formalize how to show a schedule is optimal.

\subsection{Complete Graphs}
\begin{wrapfigure}{r}{0.5\textwidth}
	\vspace{-15pt}
	\begin{framed}
		\centering
		\includegraphics[scale=.15]{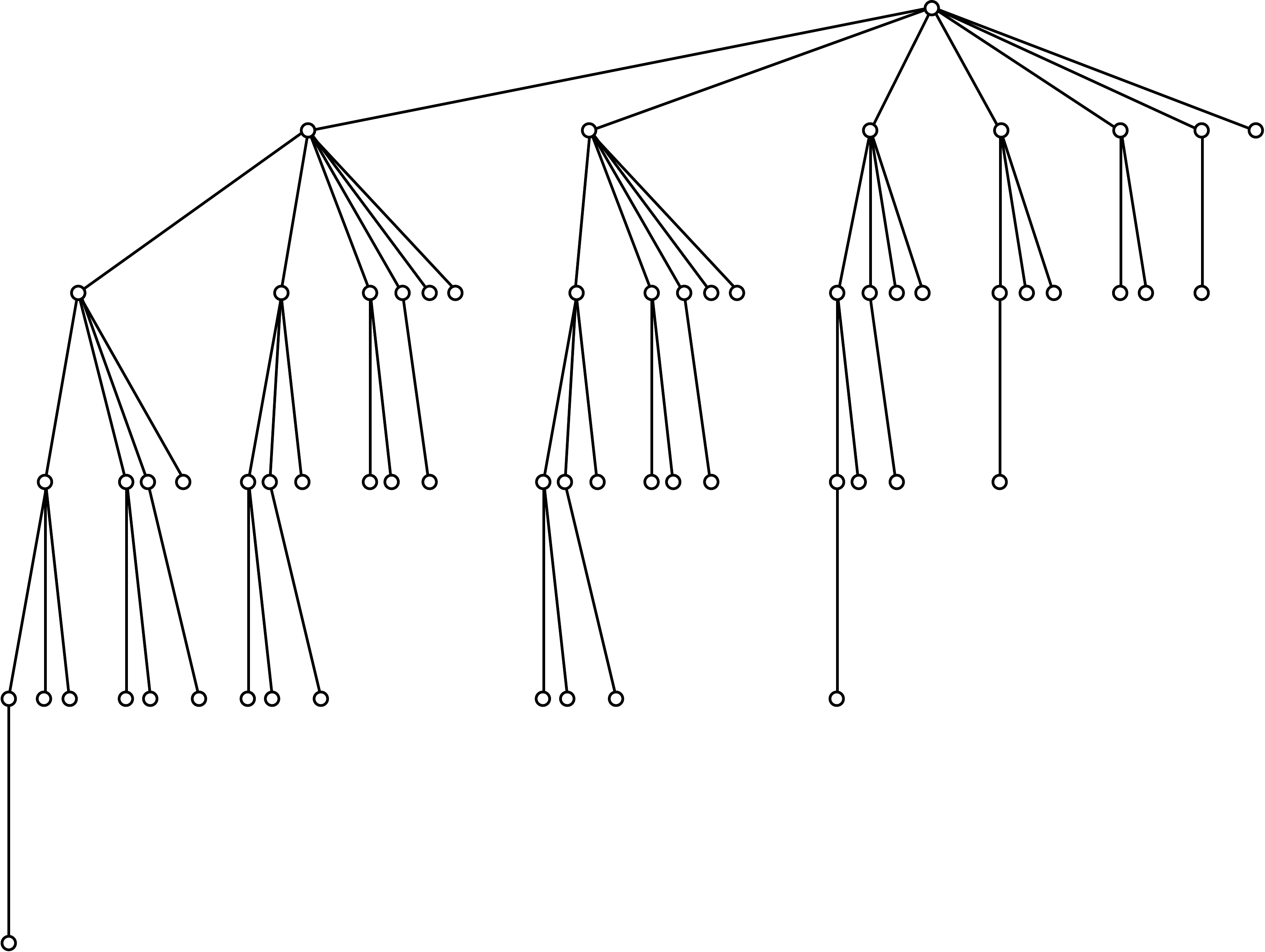}
		\caption{$T(16)$ for $t_c = 2$, $t_m = 1$.}
		\label{fig:optTreeConstruct}
	\end{framed}
	\vspace{-5pt}
\end{wrapfigure}
We now describe our optimal polynomial-time algorithm for complete graphs. 
This algorithm produces a schedule which greedily aggregates on a particular tree, $T^*_n$. In order to describe this tree, we first introduce the tree $T(R, t_c, t_m)$. This tree can be thought of as the largest tree for which greedy aggregation aggregates all \wads in $R$ rounds given computation cost $t_c$ and communication cost $t_m$. We will overload notation and let $T(R)$ denote $T(R, t_c, t_m)$ for some fixed values of $t_c$ and $t_m$. Let the root of a tree be the node in that tree with no parents. Also, given a tree $T_1$ with root $r$ we define $T_1 \textsc{ join } T_2$ as $T_1$ but where $r$ also has $T_2$ as an additional subtree. We define $T(R)$ as follows (see \Cref{fig:optTreeConstruct} for an example):
\begin{align*}
T(R) \coloneqq \begin{cases} \text{A single leaf} &\text{if $R < t_m + t_c$}\\
T(R - t_c) \textsc{ join } T(R - t_c - t_m) & \text{otherwise} \end{cases}
\end{align*}
%\aknote{Note that this construction yields a tree with a clearly designated root, $r$.}
%We give an example of $T(16)$ for $t_c = 2$ and $t_m = 1$ in \Cref{fig:optTreeConstruct}. 
%Notice that for $t_c = t_m = 1$ the tree $T(R)$ is just a Fibonacci tree \cite{horibe1983notes} of recursion depth $R$; $T(R)$, then, can be thought of as a generalization of Fibonacci trees.\footnote{This echoes the fact that only graphs containing binomial heaps have $\log n$-length telephone broadcasts \cite{elkin2005combinatorial}.}

%\begin{figure}
%\centering
%\includegraphics[scale=.15]{figs/optGraphStructure.pdf}
%\caption{An example of $T(16)$ where $t_c = 2$, $t_m = 1$.}
%\label{fig:optTreeConstruct}
%\end{figure}

Since an input to the \gump problem consists of $n$ nodes, and not a desired number of rounds, we define $R^*(n, t_c, t_m)$ to be the minimum value such that $|T(R^*(n, t_c, t_m))| \geq n$. We again overload notation and let $R^*(n)$ denote $R^*(n, t_c, t_m)$. Formally,
\begin{align*}
R^*(n) \coloneqq \min \{R : |T(R)| \geq n\}.
\end{align*}
We let $T^*_n$ denote $T(R^*(n))$. For ease of presentation we assume that $|T^*_n| = n$.\footnote{If $|T^*_n| > n$, then we could always ``hallucinate'' extra nodes where appropriate.}

The schedule produced by our algorithm will simply perform greedy aggregation on $T^*_n$. We now formally define greedy aggregation and establish its runtime on the tree $T(R)$.

\begin{definition}[Greedy Aggregation]
	\label{def:greedyaggregation}
	\emph{Given an $r$-rooted tree, let the \emph{greedy aggregation} schedule be defined as follows. In the first round, every node except for $r$ sends its \wad to its parent. In subsequent rounds we do the following. If a node is not busy and has at least two \wads, it performs a computation. If a non-root node is not busy, has exactly one \wad, and has received a \wad from every child in previous rounds, it forwards its \wad to its parent.} 
\end{definition}

%\begin{lemma}
%	\label{lem:greedyschedulelength}
%	Greedy aggregation on $T(R)$ terminates in $R$ rounds.
%\end{lemma}

\begin{restatable}{lemma}{greedyschedulelength}
	\label{lem:greedyschedulelength}
	Greedy aggregation on $T(R)$ terminates in $R$ rounds.
\end{restatable}

\begin{proof}
	We will show by induction on $k \geq 0$ that greedy aggregation results in the root of $T(k)$ having a \wad of size $|T(k)|$ after $k$ rounds. The base cases of $k \in [0, t_m + t_c)$ are trivial, as nothing needs to be combined. For the inductive step, applying the inductive hypothesis and using the recursive structure of our graph tells us that the root of $T(k + t_c)$ has a \wad of size $|T(k)|$ at its root in $k$ rounds, and the root of the child $T(k - t_m)$ has a \wad of size $|T(k - t_m)|$ at its root in $k - t_m$ rounds. Therefore, by the definition of greedy aggregation, the root of $T(k - t_m)$ sends its \wad of size $|T(k - t_m)|$ to the root of $T(k + t_c)$ at time $k - t_m$, which means the root of $T(k + t_c)$ can compute a \wad of size  $|T(k - t_m)| + |T(k)| = |T(k + t_c)|$ by round $k + t_c$.
\end{proof}

To build intuition about how quickly $T_n^*$ grows, see \Cref{fig:optSchedLength} for an illustration of $|T_n^*|$ as a function of $n$ for specific values of $t_c$ and $t_m$.
Furthermore, notice that $T(R)$ and $T^*_n$ are constructed in such a way that greedy aggregation pipelines computation and communication. We can now formalize our optimal algorithm, which simply outputs the greedy aggregation schedule on $T^*_n$, as Algorithm~\ref{alg:optComplete}. The following theorem is our main result for this section. 

\begin{algorithm}[h]
	\caption{\textsc{OptComplete}($t_c$, $t_m$, $n$)}
	\label{alg:optComplete}
	\begin{algorithmic}
		\Statex \textbf{Input:} $t_c$, $t_m$, $n$
		\Statex \textbf{Output:} A schedule for \gump on $K_n$ with parameters $t_c$ and $t_m$
		\State Arbitrarily embed $T_n^*$ into $K_n$
		\State \Return Greedy aggregation schedule on $T_n^*$ embedded in $K_n$
	\end{algorithmic}
\end{algorithm}

\begin{restatable}{theorem}{thmopt}
	\label{thm:opt}
	Given a complete graph $K_n$ on $n$ vertices and any $t_m, t_c \in \mathbb Z^+$, \textsc{OptComplete} optimally solves \gump on the \gum $(K_n, t_c, t_m)$ in polynomial time. %The schedule produced by \textsc{OptComplete} is the  greedy aggregation schedule on $T^*_n$.
\end{restatable}

To show that \Cref{thm:opt} holds, we first note that \textsc{OptComplete} trivially runs in polynomial time. Therefore, we focus on showing that greedy aggregation on $T^*_n$ optimally solves the \gump problem on $K_n$. 
We demonstrate this claim by showing that, given $R$ rounds, $|T(R)|$ is the size of the largest solvable graph. Specifically, we will let $N^*(R)$ be the size of the largest complete graph on which one can solve \gump in $R$ rounds, and we will argue that $N^*(R)$ obeys the same recurrence as $|T(R)|$.

%\begin{definition}
%	\emph{Let $N^*(R) \in \mathbb{Z}^+$ be the number of nodes in the maximum size complete graph on which one can solve the \gump problem in at most $R$ total rounds.} %Equivalently, $N^*(R)$ is the size of the largest \wad produced by a schedule of length $R$ in a sufficiently large complete graph.
%\end{definition}

%We now sketch how we show that $|T(R)| = N^*(R)$ by showing that $|T(R)|$ and $N^*(R)$ satisfy the same recurrence; full proofs deferred to \Cref{appsubsec:thmopt}.

First notice that the base case of $N^*(R)$ is trivially $1$.
\begin{lemma}\label{lem:NStarBC}
	For $R \in \mathbb{Z}_0^+$ we have that $N^*(R) = 1$ for $R < t_c + t_m$.
\end{lemma}
\begin{proof}
	If $R < t_c + t_m$, there are not enough rounds to send and combine a \wad, and so the \gump problem can only be solved on a graph with one node.
\end{proof}

We now show that for the recursive case $N^*(R)$ is always at least as large as $N^*(R - t_c) + N^*(R - t_c - t_m)$, which is the recurrence that defines $|T(R)|$.  

\begin{lemma}\label{lem:NStarLB}
	For $R \in \mathbb{Z}_0^+$ we have that $N^*(R) \geq N^*(R - t_c) + N^*(R - t_c - t_m)$ for $R \geq t_c + t_m$.
\end{lemma}
\begin{proof}
	Suppose $R \geq t_c + t_m$. Let $S_1$ be the optimal schedule on the complete graph of $N^*(R - t_c)$ nodes with terminus $v_{t1}$ and let $S_2$ be the optimal schedule on the complete graph of size $N^*(R - t_c - t_m)$ with corresponding terminus $v_{t2}$. Now consider the following solution on the complete graph of $N^*(R - t_c) + N^*(R - t_c - t_m)$ nodes. Run $S_1$ and $S_2$ in parallel on $N^*(R - t_c)$ and $N^*(R - t_c - t_m)$ nodes respectively, and once $S_2$ has completed, forward the \wad at $v_{t2}$ to $v_{t1}$ and, once it arrives, have $v_{t1}$ perform one computation. This is a valid schedule which takes $R$ rounds to solve \gump on $N^*(R - t_c) + N^*(R - t_c - t_m)$ nodes. Thus, we have that $N^*(R) \geq N^*(R - t_c) + N^*(R - t_c - t_m)$ for $R \geq t_c + t_m$.
\end{proof}

It remains to show that this bound on the recursion is tight. To do so, we case on whether $t_c \geq t_m$ or $t_c < t_m$.
When $t_c \geq t_m$, we perform a straightforward case analysis to show that $N^*$ follows the same recurrence as $T^*_n$. Specifically, we case on when the last token in the optimal schedule was created to show the following.

\begin{restatable}{lemma}{opttc} \label{lem:opttc}
	When $t_c \geq t_m$ for $R \in \mathbb{Z}_0^+$ it holds that $N^*(R) = N^*(R - t_c) + N^*(R - t_c - t_m)$.
\end{restatable}

\begin{proof}
	Suppose that $R \geq t_c + t_m$. By \Cref{lem:NStarLB}, it is sufficient to show that $N^*(R) \leq N^*(R - t_c) + N^*(R - t_c - t_m)$. Consider the optimal solution given $R$ rounds. The last action performed by any node must have been a computation that combines two \wads, $a$ and $b$, at the terminus $v_t$ because, in an optimal schedule, any further communication of the last \wad increases the length of the schedule. We now consider three cases. \begin{itemize}

		\item In the first case, $a$ and $b$ were both created at $v_t$. Because both of $a$ or $b$ could not have been created at time $R - t_c$, one of them must have been created at time $R - 2t_c$ at the latest. This means that $N^*(R) \leq N^*(R - t_c) + N^*(R - 2t_c) \leq  N^*(R - t_c) + N^*(R - t_c - t_m)$.
		
		\item In the second case, exactly one of $a$ or $b$ (without loss of generality, $a$) was created at $v_t$. This means that $b$ must have been sent to $v_t$ at latest at time $R - t_c - t_m$. It follows that $N^*(R) \leq N^*(R - t_c) + N^*(R - t_c - t_m)$.
		
		\item In the last case, neither $a$ nor $b$ was created at $v_t$. This means that both must have been sent to $v_t$ at the latest at time $R - t_c - t_m$. We conclude that $N^*(R) \leq N^*(R - t_c - t_m) + N^*(R - t_c - t_m) \leq N^*(R - t_c) + N^*(R - t_c - t_m)$.
	\end{itemize}
	
	Thus, in all cases we have $N^*(R) \leq N^*(R - t_c) + N^*(R - t_c - t_m)$.
	% Lastly, $N^*(R) \geq N^*(R - t_c) + N^*(R - t_c - t_m)$ follows immediately from \Cref{lem:NStarLB} and so we conclude that for $R \geq t_c + t_m$ it holds that $N^*(R) = N^*(R - t_c) + N^*(R - t_c - t_m)$ when $t_c \geq t_m$.
	%In each case, because we know that $t_c \geq t_m$, it is clear that the desired inequality follows.
\end{proof}

We now consider the case in which communication is more expensive than computation, $t_c < t_m$. One might hope that the same case analysis used when $t_c \geq t_m$ would prove the desired result for when $t_c < t_m$. However, we must do significantly more work to show that $N^*(R) = N^*(R - t_c) + N^*(R - t_c - t_m)$ when $t_c < t_m$. We do this by establishing structure on the schedule which solves \gump on $K_{N^*(R)}$ in $R$ rounds: we successively modify an optimal schedule in a way that does not affect its validity or length but which adds structure to the schedule. 

Specifically, we leverage the following insights\,---\,illustrated in \Cref{fig:modInsights}\,---\,to modify schedules. \emph{Combining insight}: Suppose node $v$ has two \wads in round $t$, $a$ and $b$, and $v$ sends $a$ to node $u$ in round $t$. Node $v$ can just aggregate $a$ and $b$, treat this aggregation as it treats $b$ in the original schedule and $u$ can just pretend that it receives $a$ in round $t+t_m$. That is, $u$ can ``hallucinate'' that it has \wad $a$. Note that this insight crucially leverages the fact that $t_c < t_m$, since otherwise the performed computation would not finish before round $t + t_m$. \emph{Shortcutting insight}: Suppose node $v$ sends a \wad to node $u$ in round $t$ and node $u$ sends a \wad to node $w$ in a round in $[t, t + t_m]$. Node $v$ can ``shortcut'' node $u$ and send to $w$ directly and $u$ can just not send.

\begin{figure}[t]
	\centering
	\begin{subfigure}{.49\textwidth}
		\includegraphics[width=\textwidth]{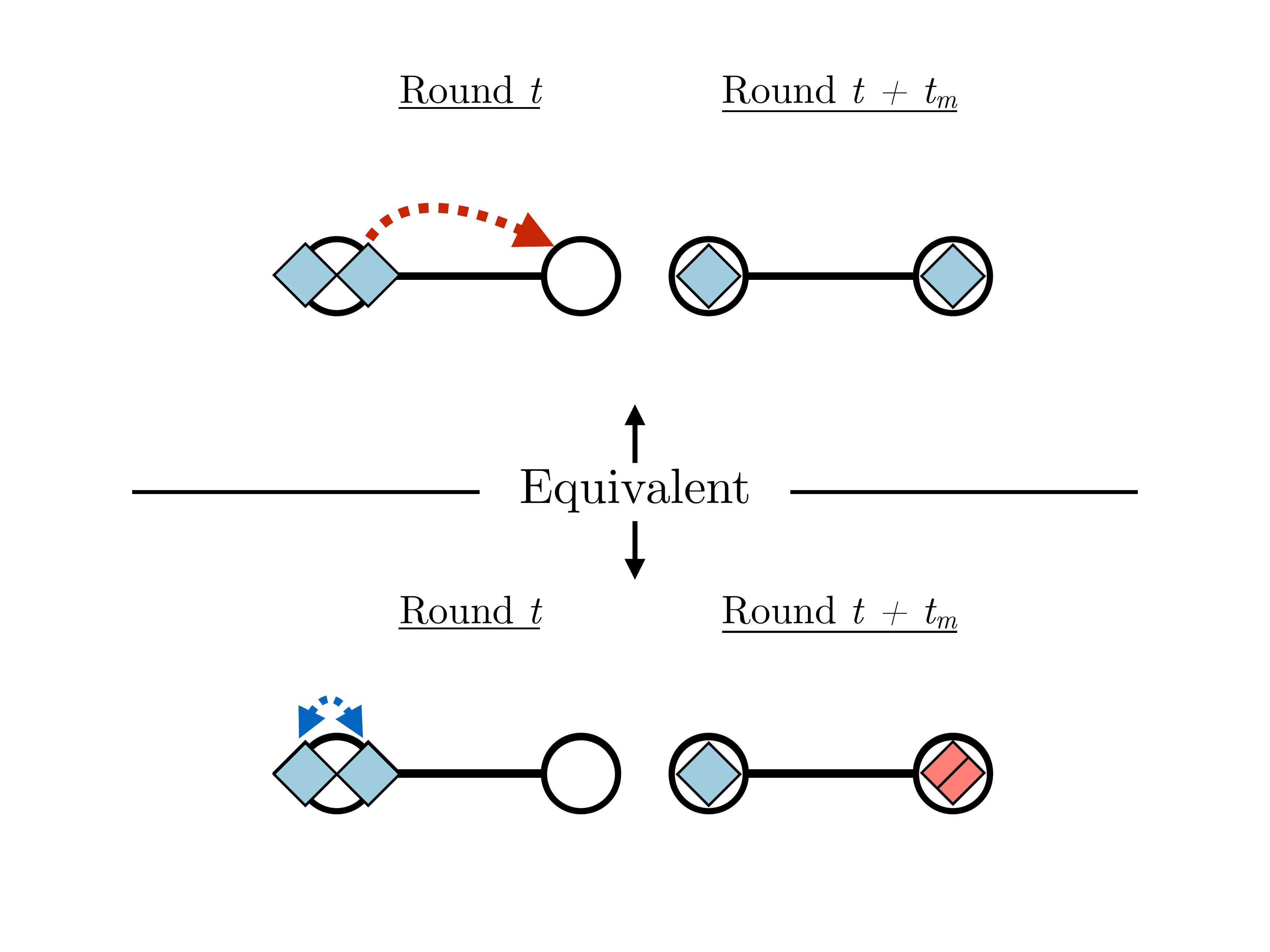}
		\caption{Combining insight}
	\end{subfigure}
	\begin{subfigure}{.49\textwidth}
		\includegraphics[width=\textwidth]{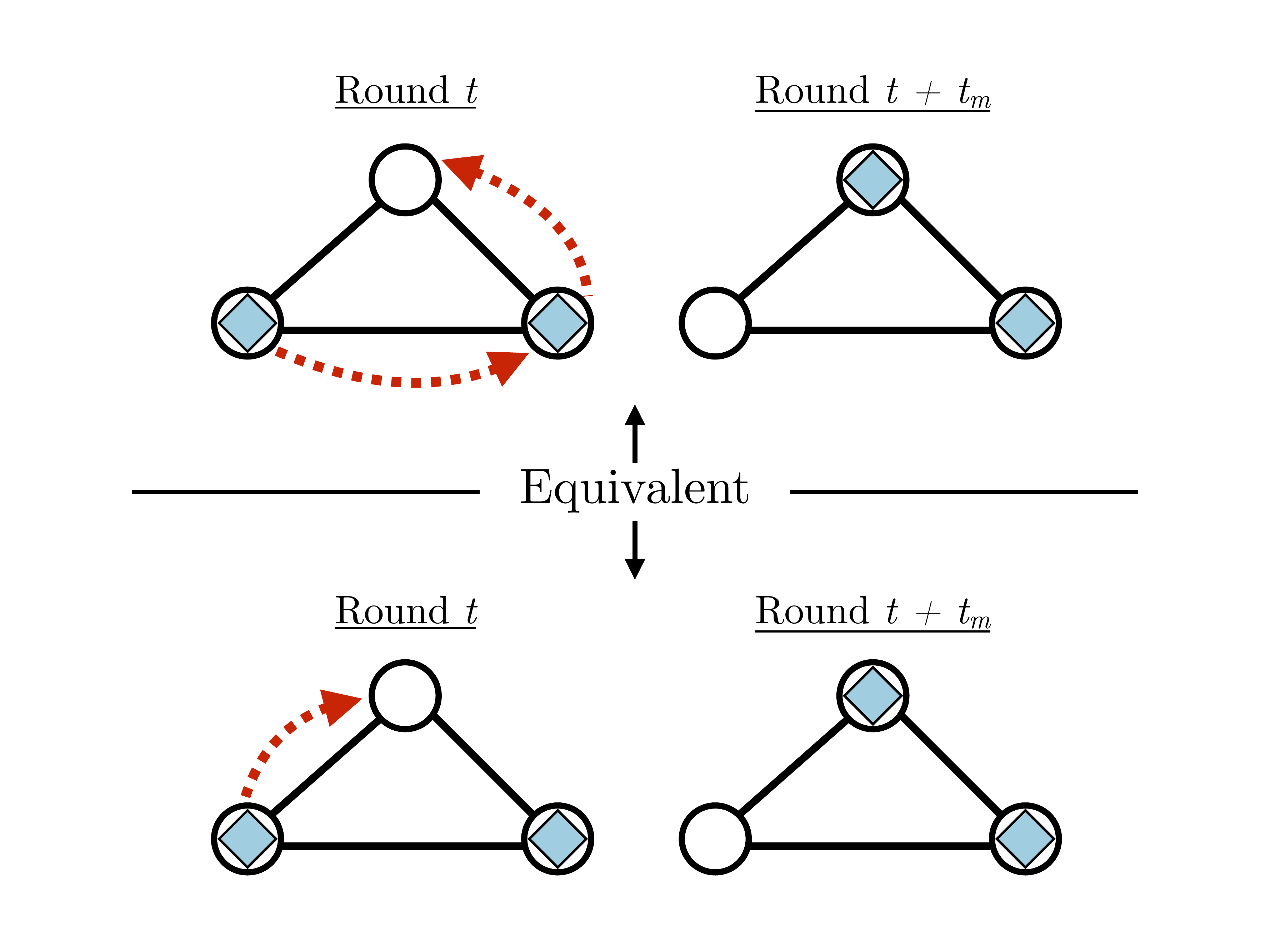}
		\caption{Shortcutting insight}
	\end{subfigure}
	
	\caption{An illustration of the shortcutting and combining insights. Here, \wads are denoted by blue diamonds, and hallucinated \wads are denoted by striped red diamonds. As before, a   red arrow from node $u$ to node $v$ means that $u$ sends to $v$, and a double-ended blue arrow between two \wads $a$ and $b$ means that $a$ and $b$ are combined at the node. Notice that which nodes have \wads and when nodes have \wads are the same under both modifications (though in the combining insight, a node is only hallucinating that it has a \wad).}
	\label{fig:modInsights}
\end{figure}

Through modifications based on these insights we show that there exists an optimal schedule where the last node to perform a computation  never communicates, and every computation performed by this node computes on the token with which this node started. This structure, in turn, allows us to establish the following lemma, which asserts that when $t_c < t_m$ we have that $N^*(R)$ and $|T(R)|$ follow the same recurrence.
\begin{restatable}{lemma}{opttm}
	\label{lem:opttm}
	When $t_c < t_m$, for $R \in \mathbb{Z}_0^+$ it holds that $N^*(R) = N^*(R - t_c) + N^*(R - t_c - t_m)$.
\end{restatable}

The proof of the lemma is relegated to Appendix~\ref{appsubsec:thmopt}. We are now ready to prove the theorem. 

%Lemmas \ref{lem:opttc} and \ref{lem:opttm} establish that the recursive structure of $N^*(R)$ is the same as that of $|T(R)|$ for both cases. Thus, we conclude that $|T(R)|$ is the most nodes one can hope to solve \gump on in a complete graph in $R$ rounds and so $T^*_n$ is optimal, yielding \Cref{thm:opt}.

\begin{proof}[Proof of \Cref{thm:opt}]
	On a high level, we argue that the greedy aggregation schedule on $T(R)$ combines $N^*(R)$ nodes in $R$ rounds and is therefore optimal. Combining \Cref{lem:NStarBC}, \Cref{lem:opttc}, and \Cref{lem:opttm} we have the following recurrence on $N^*(R)$ for $R \in \mathbb{Z}_0^+$.
	
	\begin{align*}
	N^*(R) = \begin{cases}
	1 & \text{if $R < t_c + t_m $}\\
	N^*(R - t_c) + N^*(R - t_c - t_m) & \text{if $R \geq t_c + t_m$}
	\end{cases}
	\end{align*}
	
	Notice that this is the recurrence which defines $|T(R)|$ so for $R \in \mathbb{Z}_0^+$ we have that $N^*(R)=|T(R)|$, and by \Cref{lem:greedyschedulelength}, the greedy aggregation schedule on $T(R)$ terminates in $R$ rounds.
	
	Thus, the greedy aggregation schedule on $T(R)$ solves \gump on $K_{|T(R)|} = K_{N^*(R)}$ in $R$ rounds, and therefore is an optimal solution for $K_{N^*(R)}$.  Since $T_n^*$ is the smallest $T(R)$ with at most $n$ nodes, greedy aggregation on $T_n^*$ is optimal for $K_n$ and so \textsc{OptComplete} optimally solves \gump on $K_n$. Finally, the polynomial runtime is trivial.
\end{proof} 

\section{Hardness and Approximation for Arbitrary Graphs}\label{sec:arbgraphs}

We now consider the \gump problem on arbitrary graphs. Unlike in the case of complete graphs, the problem turns out to be computationally hard on arbitrary graphs. The challenge in demonstrating the hardness of \gump is that the optimal schedule for an arbitrary graph does not have a well-behaved structure. Our insight here is that by forcing a single node to do a great deal of computation we can impose structure on the optimal schedule in a way that makes it reflect the minimum dominating set of the graph. The following theorem formalizes this; its full proof is relegated to Appendix~\ref{app:hard}.

\begin{restatable}{theorem}{hardapx}
	\label{thm:hardapx}
	\gump cannot be approximated by a polynomial-time algorithm within $(1.5 - \eps)$ for $\eps \geq \frac{1}{o(\log n)}$ unless $\text{P} = \text{NP}$.
\end{restatable}

Therefore, our focus in this section is on designing an approximation algorithm. Specifically, we construct a polynomial-time algorithm, \solvegump, which produces a schedule that solves \gump on arbitrary graphs using at most $O(\log n \cdot \log \frac{\OPT}{t_m})$ multiplicatively more rounds than the optimal schedule, where $\OPT$ is the length of the optimal schedule. Define the diameter $D$ of graph $G$ as $\max_{v,u} d(u,v)$. Notice that $\OPT/t_m$ is at most $(n-1)t_c/t_m + D$ since $\OPT \leq (n-1)(t_c + D \cdot t_m)$: the schedule that picks a pair of nodes, routes one to the other then aggregates and repeats $n-1$ times is valid and takes $(n-1)(t_c + D \cdot t_m)$ rounds. Thus, our algorithm can roughly be understood as an $O(\log^2 n)$ approximation algorithm. Formally, our main result for this section is the following theorem whose lengthy proof we summarize in the rest of this section.

\begin{restatable}{theorem}{apxthm}
	\label{thm:apxalgo}
	\solvegump is a polynomial-time algorithm that gives an $O(\log n \cdot \log \frac{\OPT}{t_m})$-approximation for \gump with high probability.
\end{restatable}

The rest of this section provides an overview of this theorem's lengthy proof. %The full proof can be found in Appendix~\ref{app:approx}.
Our approximation algorithm, \solvegump, is given as \Cref{alg:solveGumCen}. \solvegump performs $O(\log n)$ repetitions of: designate some subset of nodes with \wads sinks and the rest of the nodes with \wads sources; route \wads at sources to sinks. If $t_c > t_m$, we will delay computations until \wads from sources arrive at sinks, and if $t_m \geq t_c$, we will immediately aggregate \wads that arrive at the same node.

\subsection{\gump Extremes (Warmup)}\label{sec:ToyEgs}
Before moving on to a more technical overview of our algorithm, we build intuition by considering two extremes of \gump. 

\paragraph{$t_m \ll t_c$.} First, consider the case where $t_m \ll t_c$; that is, communication is very cheap compared to computation. As computation is the bottleneck here, we can achieve an essentially optimal schedule by parallelizing computation as much as possible. That is, consider a schedule consisting of $O(\log n)$ repetitions of: (1) each node with a \wad uniquely pairs off with another node with a \wad; (2) one node in each pair routes its \wad to the other node in its pair; (3) nodes that received a \wad perform one computation. This takes $O(t_c \cdot \log n )$ rounds to perform computations along with some amount of time to perform communications. But, any schedule takes at least $\Omega( t_c \cdot \log n)$ rounds, even if communication were free and computation were perfectly parallelized. Because the time to perform communications is negligible, this schedule is essentially optimal.

\paragraph{$t_c \ll t_m$.} Now consider the case where $t_c \ll t_m$; that is, computation is very cheap compared to communication. In this setting, we can provide an essentially optimal schedule by minimizing the amount of communication that occurs. In particular, we pick a center $c$ of the graph\footnote{The center of graph $G$ is $\argmin_v \max_{u} d(v, u)$ where $d(v, u)$ is the length of the shortest $u-v$ path.} and have all nodes send their \wads along the shortest path towards $c$. At any point during this schedule, it is always more time efficient for a node with multiple \wads to combine its \wads together before forwarding them since $t_c \ll t_m$. Thus, if at any point a node has multiple \wads, it combines these into one \wad and forwards the result towards $c$. Lastly, $c$ aggregates all \wads it receives. This schedule takes $t_m \cdot r$ time to perform its communications, where $r$ is the radius of the graph,\footnote{The radius of graph $G$ is $\min_v \max_u d(v, u)$.} and some amount of time to perform its computations.  However, because for every schedule there exists a \wad that must travel at least $r$ hops, any schedule takes at least $\Omega(r \cdot t_m)$ rounds. Computations take a negligible amount of time since $t_c \ll t_m$, which means that this schedule is essentially optimal.\\

See \Cref{fig:toyEG} for an illustration of these two schedules. Thus, in the case when $t_m \ll t_c$, we have that routing between pairs of nodes and delaying computations is essentially optimal, and in the case when $t_c \ll t_m$, we have that it is essentially optimal for nodes to greedily aggregate \wads before sending. These two observations will form the foundation of our approximation algorithm.

\begin{figure}
	\centering
	\begin{subfigure}[t]{0.49\textwidth}
		\centering
		\includegraphics[scale=0.1,page=2]{figs/toyEg.pdf}
		\caption{$t_m \ll t_c$}
	\end{subfigure}%
	~
	\begin{subfigure}[t]{0.49\textwidth}
		\centering
		\includegraphics[scale=0.1,page=1]{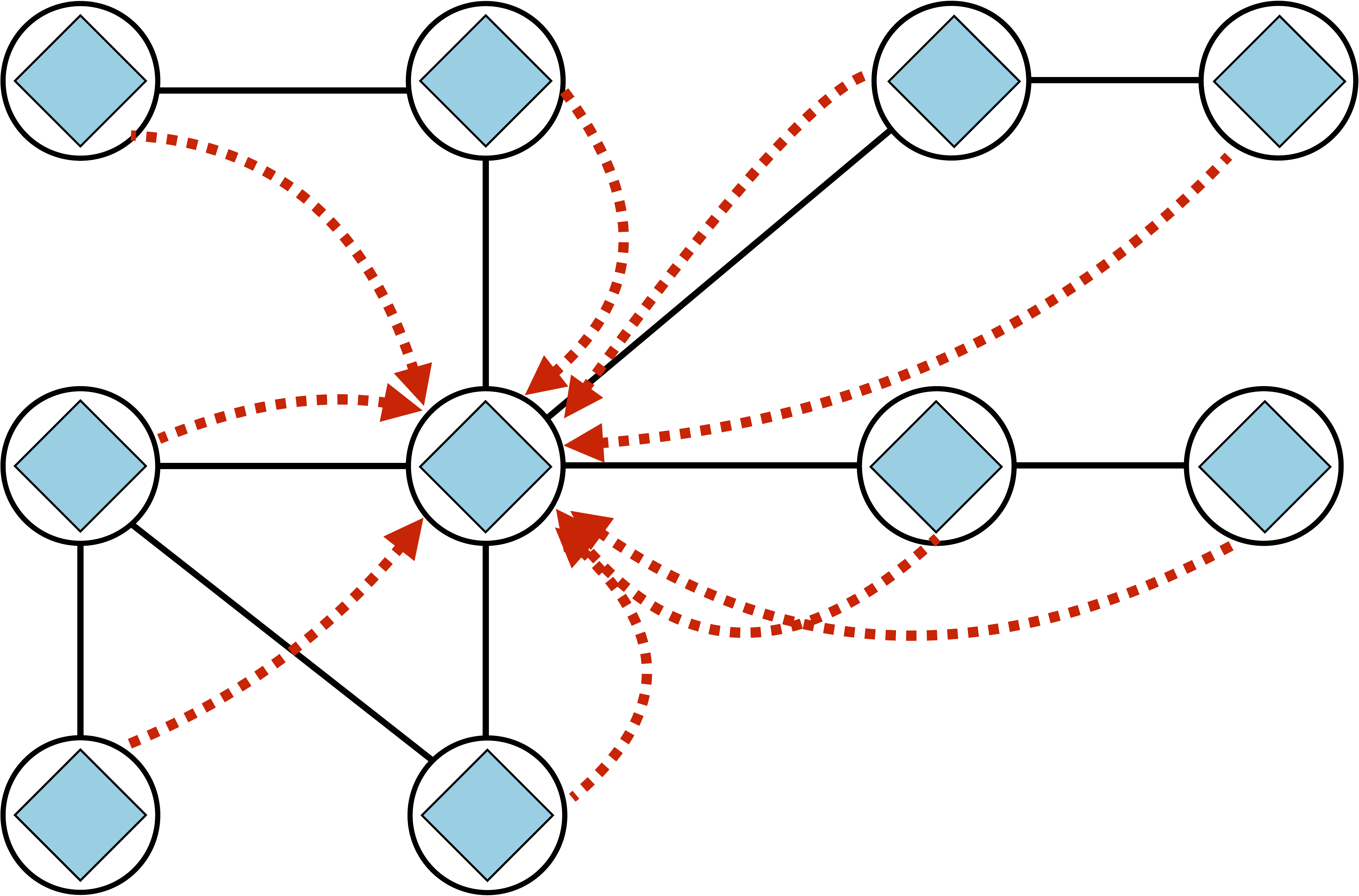}
		\caption{$t_c \ll t_m$}
	\end{subfigure}%
	\caption{An illustration of essentially optimal schedules for the extremes of the \gump problem. Dotted red arrows give the node towards which each node routes. In the case where $t_m \ll t_c$ one would repeat this sort of routing $O(\log n)$ times.}\label{fig:toyEG}
\end{figure}

\subsection{Approximation Algorithm}\label{sec:apxAlg}

Recall that our approximation algorithm routes tokens from designated sources to designated sinks $O(\log n)$ times. Formally, the problem which our algorithm solves $O(\log n)$ times is as follows.

\begin{definition}[\textsc{Route and Compute} Problem] \emph{The input to the \textsc{Route and Compute} Problem consists of a set $U \subseteq V$ and a set of directed paths $\vec{\mcP_U} = \{\vec{P_u} : u \in U\}$ where: (1) $u \in U$ has a \wad and is the source of $\vec{P_u}$; (2) every sink of every path $\vec{P_u}$ has a \wad ; (3) if $u$ and $t_u$ are the sources and sinks of $\vec{P_u} \in \vec{\mcP_U}$, respectively, then neither $u$ nor $t_u$ are endpoints of any $\vec{P_{u'
		}} \in \vec{\mcP_U}$ for $u' \neq u$. A solution of cost $C$ is a schedule of length $C$ which, when run, performs computations on a constant fraction of \wads belonging to nodes in $U$.}
\end{definition}

\solvegump repeatedly calls a subroutine, \textsc{GetDirectedPaths}, to get a set of paths for which it would like to solve the \textsc{Route and Compute} Problem. It then solves the \textsc{Route and Compute} Problem for these paths, using $\textsc{RoutePaths}_m$ if $t_c \leq t_m$ or $\textsc{RoutePaths}_c$ if $t_c > t_m$. Below we give an overview of these procedures. The proofs of the lemmas in this section, as well as further details regarding \solvegump, are relegated to Appendix~\ref{app:approx}.
\begin{algorithm}
	\caption{\solvegump}
	\label{alg:solveGumCen}
	\begin{algorithmic}
		\Statex \textbf{Input: }\gump instance given by graph $G = (V, E), t_c, t_m$
		\Statex \textbf{Output:} A schedule for the input \gump problem
		\State $W \gets V$
		\For{iteration $i \in O(\log n)$}
		\State $\vec{\mcP_U} \gets \textsc{GetDirectedPaths}(W, G)$
		\State \algorithmicif\ $t_c > t_m$ \algorithmicthen\ $\textsc{RoutePaths}_m(\vec{\mcP_U})$
		\State \algorithmicif\ $t_c \leq t_m$ \algorithmicthen\ $\textsc{RoutePaths}_c(\vec{\mcP_U})$ 
		\State $W \gets \{v : \text{$v$ has 1 \wad}\}$
		\EndFor
	\end{algorithmic}
\end{algorithm}

\subsubsection{Producing Paths on Which to Route}\label{sec:prodPaths}
We now describe \textsc{GetDirectedPaths}. First, for a set of paths $\mcP$, we define the \emph{vertex congestion} of $\mcP$ as $\text{con}(\mcP) = \max_v \sum_{P \in \mcP}(\text{\# occurences of $v \in P$})$, and the \emph{dilation} of $\mcP$ as $\max_{P \in \mcP} |P|$. 

Given that nodes in $W \subseteq V$ have \wads, \textsc{GetDirectedPaths} solves a flow LP which has a flow for each $w \in W$ whose sinks are $w' \in W$ such that $w' \neq w$. The objective of this flow LP is the vertex congestion. The flow for each $w \in W$ defines a probability distribution over (undirected) paths with endpoints $w$ and $w'$ where $w' \neq w$ and $w' \in W$. Given these probability distributions, we repeatedly sample paths by taking random walks proportional to LP values of edges until we produce a set of paths\,---\,one for each $w \in W$\,---\,with low vertex congestion. Lastly, given our undirected paths, we apply another subroutine to direct our paths and fix some subset of nodes $U \subset W$ as sources such that $|U|$ is within a constant fraction of $|W|$. The key property of the LP we use is that it has an optimal vertex congestion comparable to $\OPT$, the length of the optimal \gump schedule. Using this fact and several additional lemmas we can prove the following properties of \textsc{GetDirectedPaths}. 
\begin{restatable}{lemma}{getpaths}\label{lem:getPaths}
	Given $W \subseteq V$, \textsc{GetDirectedPaths} is a randomized polynomial-time algorithm that returns a set of directed paths, $\vec{\mcP_U} = \{P_u : u \in U\}$ for $U \subseteq W$, such that with high probability at least $1/12$ of nodes in $W$ are sources of paths in $\vec{\mcP_U}$ each with a unique sink in $W$. Moreover, 
	$$\con(\vec{\mcP_U}) \leq O\left(\frac {\OPT}{\min(t_c, t_m)} \log \frac{\OPT}{t_m} \right) \text{ and } \dil(\vec{\mcP_U}) \leq \frac{8 \OPT}{t_m}.$$
\end{restatable}

%Details on our randomized rounding and how we direct our paths are deferred to Appendix \ref{app:approx}.
%We provide formal pseudocode for \textsc{GetDirectedPaths}, additional details regarding properties of our LP, details about subroutines for randomized rounding and directing paths, and a proof of properties of \textsc{GetDirectedPaths} in Appendix~\ref{app:approx}. 

%Having formally defined our LP, we give pseudocode for \textsc{GetDirectedPaths} in \Cref{alg:getDirPaths} where our randomized rounding algorithm (\textsc{SampleLPPaths}) and algorithm for fixing the direction of paths (\textsc{AssignPaths}) are described in \Cref{app:approx}. $R \coloneqq \lceil 2(n-1) \cdot(t_c + D \cdot t_m) / t_m \rceil$ is the range over which we search for $L^*$. Making use of the previous lemma, we can demonstrate the following crucial properties of \textsc{GetDirectedPaths}. 
%
%\begin{algorithm} 
%	\caption{\textsc{GetDirectedPaths}(G, W)}
%	\label{alg:getDirPaths}
%	\begin{algorithmic}
%		\Statex \textbf{Input:} $W \subseteq V$ where $w \in W$ has a \wad
%		\Statex \textbf{Output:} Directed paths between nodes in $W$
%		\State $L \gets \argmin_{\hat{L} \in [R]} \left[ t_m \cdot \hat{L} + \min(t_c, t_m) \cdot t(\hat{L}) \right]$
%		\State $f_w^* \gets \textsc{PathsFlowLP$(L)$}$
%		\State $\mcP_W \gets \textsc{SampleLPPaths}(f_w^*, L, W)$
%		\State $\vec{\mcP_U} \gets \textsc{AssignPaths}(\mcP_W, W)$
%		\State \Return $\vec{\mcP_U}$
%	\end{algorithmic}
%\end{algorithm}

\subsubsection{Routing Along Produced Paths}\label{sec:routingAlongProdPaths}
We now specify how we route along the paths produced by \textsc{GetDirectedPaths}. 
%In particular, if $t_c \leq t_m$ we use one subroutine, $\textsc{RoutePaths}_c$, and if $t_c > t_m$ we use another subroutine, $\textsc{RoutePaths}_m$.
If $t_c > t_m$, we run $\textsc{RoutePaths}_m$ to delay computations until \wads from sources arrive at sinks, and if $t_m \geq t_c$, we run $\textsc{RoutePaths}_c$ to immediately aggregate \wads that arrive at the same node. 

\paragraph{Case of $t_c > t_m$}

%In this section we provide a polytime algorithm, $\textsc{RoutePaths}_m$, that we use to solve the \textsc{Route and Compute} Problem when $t_c > t_m$. 
$\textsc{RoutePaths}_m$ adapts the routing algorithm of Leighton et al.~\cite{leighton1994packet}\,---\,which was simplified by Rothvo\ss~\cite{rothvoss2013simpler}\,---\,to efficiently route from sources to sinks.\footnote{Our approach for the case when $t_c > t_m$ can be simplified using techniques from Srinivasan and Teo~ \cite{srinivasan2001constant}. In fact, using their techniques we can even shave the $\frac{\log \OPT }{t_c}$ factor in our approximation. However, because these techniques do not take computation into account, they do not readily extend to the case when $t_c \leq t_m$. Thus, for the sake of a unified exposition, we omit the adaptation of their results.} We let $\textsc{OPTRoute}$ be this adaptation of the algorithm of Leighton et al.~\cite{leighton1994packet}. %and describe it and its properties now.

\begin{restatable}{lemma}{optroute}
	\label{lem:OPTRoute}
	Given a set of directed paths $\vec{\mcP_U}$ with some subset of endpoints of paths in $\vec{\mcP_U}$ designated sources and the rest of the endpoints designated sinks, \textsc{OPTRoute} is a randomized polynomial-time algorithm that w.h.p.\ produces a \gum schedule that sends from all sources to sinks in $O(\text{con}(\vec{\mcP_U}) + \dil(\vec{\mcP_U}))$.
\end{restatable}

Given $\vec{\mcP_U}$, $\textsc{RoutePaths}_m$ is as follows. Run $\textsc{OPTRoute}$ and then perform a single computation. As mentioned earlier, this algorithm delays computation until all \wads have been routed.

\begin{restatable}{lemma}{routepathsm}
	\label{lem:routePathsM}
	$\textsc{RoutePaths}_m$ is a polynomial-time algorithm that, given $\vec{\mcP_U}$, solves the \textsc{Route and Compute} Problem w.h.p.\ using $O(t_m (\con(\vec{\mcP_U}) +  \dil(\vec{\mcP_U}))+ t_c)$ rounds.
\end{restatable}
%\routepathsm*
%\begin{proof}
%By \Cref{lem:OPTRoute}, \textsc{OPTRoute} takes $t_m (\con(\vec{\mcP_U}) +  \dil(\vec{\mcP_U}))$ rounds to route all sources to sinks. All sources are combined with sinks in the following computation and so $\textsc{RoutePaths}_m$ successfully solves the \textsc{Route and Compute} Problem since every source has its \wad combined with another \wad. The polynomial runtime of the algorithm is trivial.
%\end{proof}
%
%
\paragraph{Case of $t_c \leq t_m$}
%We now present $\textsc{RoutePaths}_c$, the \textsc{Route and Compute} algorithm we use when $t_c \leq t_m$. 
%The main intuition we would like to formalize is that if $t_c \leq t_m$ then greedily combining \wads is always advantageous; that is, any time a node has two \wads it is always better for the node to combine these two \wads than to send just one of them.
Given directed paths $\vec{\mcP_U}$, $\textsc{RoutePaths}_c$ is as follows. 
Initially, every sink is \emph{asleep} and every other node is \emph{awake}. For $O(\dil(\vec{\mcP_U}) \cdot t_m)$ rounds we repeat the following: if a node is not currently sending and has exactly one \wad then it forwards this \wad along its path; if a node is not currently sending and has two or more \wads then it sleeps for the remainder of the $O(\dil(\vec{\mcP_U}) \cdot t_m)$ rounds. Lastly, every node combines any \wads it has for $t_c \cdot \con(\vec{\mcP_U})$ rounds. 

\begin{restatable}{lemma}{routepathsc}\label{lem:routePathsC}
	$\textsc{RoutePaths}_c$ is a polynomial-time algorithm that, given $\vec{\mcP_U}$, solves the \textsc{Route and Compute} Problem w.h.p.\ using $O(t_c \cdot \con(\vec{\mcP_U}) + t_m \cdot \dil(\vec{\mcP_U}))$ rounds.
\end{restatable}
%\begin{proof}
%We argue that every source's \wad ends at an asleep node with at least two \wads and no more than $\con(\vec{\mcP_U})$ \wads. It follows that our computation at the end at least halves the number of \wads.
%
%First notice that if a vertex falls asleep then it will receive at most $\con(\vec{\mcP_S})$ \wads by the end of our algorithm since it is incident to at most this many paths. Moreover, notice that every \wad will either end at a sink or a sleeping vertex and every sleeping vertex is asleep because it has two or more \wads. It follows that every \wad is combined with at least one other \wad and so our schedule at least halves the total number of \wads.
%
%The length of our schedule simply comes from noting that we have $O(\dil(\vec{\mcP_U}) \cdot t_m)$ forwarding rounds followed by $\con(\vec{\mcP_U}) \cdot t_c$ rounds of computation. Thus, we get a schedule of total length $O(t_c \cdot \con(\vec{\mcP_S}) + t_m \cdot \dil(\vec{\mcP_S}))$. A polynomial runtime is trivial.
%
%\end{proof}

%\subsubsection{Putting It All Together}
%Lastly, we conclude that \solvegump solves \gump up to a multiplicative $O(\log n \cdot \log \frac{\OPT}{t_m})$ factor.

By leveraging the foregoing results, we can prove \Cref{thm:apxalgo}; see Appendix~\ref{appsubsec:apxalgoproof} for details.

\section{Future Work}
\label{app:future}
%In this work, we gave a computational treatment of several decades of sociology work on organizational structure. Specifically, we introduced the \gum model of organizations, which takes into account both the cost of computation and communication. Motivated by the tendency observed in sociology experiments wherein individuals send their values to be aggregated at a single node, we studied the \gump problem. We showed how to exactly solve the \gump problem on complete graphs in polynomial time and how to approximate it within a $O\left(\log n  \cdot \tfrac{\log \OPT}{t_m} \right)$ multiplicative factor on arbitrary graphs. We also showed that the problem is NP-hard and hard to approximate by a polynomial-time algorithm better than a $1.5$ multiplicative factor.

There are many promising directions for future work. First, as \Cref{sec:ToyEgs} illustrates, the extremes of our problem\,---\,when $t_c \ll t_m$ and when $t_m \ll t_c$\,---\,are trivial to solve. However, our hardness reduction demonstrates that for $t_c = 1$ and $t_m$ in a specific range, our problem is hard to approximate.  Determining precisely what values of $t_m$ and $t_c$ make our problem hard to approximate is open.

Next, it is not always the case that there exists a centralized coordinator to produce a schedule. We hope to give an analysis of our problem in a distributed setting as no past work in this setting takes computation into account. Even more broadly, we hope to analyze formal models of distributed computation in which nodes are not assumed to have unbounded computational resources and computation takes a non-trivial amount of time.

We also note that there is a gap between our hardness of approximation and the approximation guarantee of our algorithm. The best possible approximation, then, is to be decided by future work.

Furthermore, we are interested in studying technical challenges similar to those studied in approximation algorithms for network design. For instance, we are interested in the problem in which each edge has a cost and one must build a network subject to budget constraints which has as efficient a \gump schedule as possible.

Lastly, there are many natural generalizations of our problem. For instance, consider the problem in which nodes can aggregate an arbitrary number of \wads together, but the time to aggregate multiple \wads is, e.g., a concave function of the number of \wads aggregated. 
%Additionally, one could consider a directed version of this problem to capture settings in which it is only the case that subordinates send information to superiors to aggregate. We note that, in effect, our optimal construction behaves like a directed network, but this is not the case when solving \gump on general graphs. 
These new directions offer not only compelling theoretical challenges but may be of practical interest.%, and so we are hopeful that our problem will receive further attention.

%%
%% Bibliography
%%

%% Please use bibtex, 

\bibliography{main}

\newpage
\appendix

\section{Formal Model, Problem, and Definitions}
\label{appsec:defs}

Let us formally define the \gump problem. The input to the problem is a \gum specified by graph $G = (V,E)$ and parameters $t_c, t_m \in \mathbb{N}$. Each node starts with a single \wad.

An algorithm for this problem must provide a schedule, $S : V \times [l] \to V \cup \{ \textsf{idle}, \textsf{busy} \} $ where we refer to $|S| \coloneqq l$ as the length of the schedule. Intuitively, a schedule $S$ directs each node when to compute and when to communicate as follows:
\begin{itemize}
	\item $S(v, r) = v' \neq v$ indicates that $v$ begins passing a \wad to $v'$ in round $r$ of $S$;
	\item $S(v, r) = v$ indicates that $v$ begins combining two \wad in round $r$ of $S$;
	\item $S(v, r) = \textsf{idle}$ indicates that $v$ does nothing in round $r$;
	\item $S(v, r) = \textsf{busy}$ indicates that $v$ is currently communicating or computing.
\end{itemize}

Moreover, we define the number of computations that $v$ has performed up to round $r$ as $C_S(v, r) \coloneqq  \sum_{r' \in [r - t_c]} \mathbbm{1} (S(v, r') == v)$, the number of messages that $v$ has received up to round $r$ as $R_S(v, r) \coloneqq \sum_{r' \in [r-t_m]} \sum_{v' \neq v} \mathbbm{1} (S(v', r') == v)$, and the number of messages that $v$ has sent up to round $r$ as $M_S(v, r) \coloneqq \sum_{r' \in [r-t_m]} \sum_{v' \neq v} \mathbbm{1} (S(v, r') == v')$. Finally, define the number of \wads a node has in round $r$ of $S$ as follows.
$$
\textsf{\wads}_S(v, r) \coloneqq I(v) + R_S(v, r) - M_S(v, r) - C_S(v, r).
$$

A schedule, $S$, is \emph{valid} for \gum $(G, t_c, t_m)$ if: 
\begin{enumerate}
	\item Valid communication: If $S(v, r) = v' \neq v$ then $(v, v') \in E$, $S(v, r') = \textsf{busy}$ for $r' \in [r + 1, r + t_m]$ and $\textsf{\wads}_S(v, r) \geq 1$;
	\item Valid computation: If $S(v, r) = v$ then $S(v, r') = \textsf{busy}$ for $r' \in [r+1, r+ t_c]$ and $\textsf{\wads}_S(v, r) \geq 2$; 
	\item Full aggregation: $\sum_{v \in V} \textsf{\wads}_S(v, |S|) = 1$.
\end{enumerate}
An algorithm solves $\gump$ if it outputs a valid schedule.

\section{Deferred Related Work}\label{app:relatedWork}

%\subsection{Additional Related Theoretical Work}
%
%\paragraph{Distributed Graph Algorithms for Computing Functions}
%There is also a wealth of distributed algorithms for computing functions \cite{bawa2003estimating, kempe2003gossip, haeupler2018round,mosk2006computing,peleg2000distributed}. In this setting, there is no centralized algorithm coordinating nodes as in our setting; rather, nodes themselves are responsible for dictating their own behavior given a restricted notion of the structure of the network. The models considered in these spaces\,---\,such as the LOCAL\cite{peleg2000distributed, linial1992locality} and CONGEST \cite{peleg2000distributed} models\,---\,also assume that nodes have unbounded computation and so fail to capture how communication and computation interact as our \gum model does. We note that our algorithm for the \gump problem applies to the \gump problem where nodes communicate as in the LOCAL or CONGEST model but still take time to do computation. Although nodes can send multiple messages to neighbors in each round in these models, it is not hard to see that nodes could always send a single value when solving the \gump problem without increasing the total length of the output schedule.

%\subsection{Additional Related Applied Work}

There is a significant body of applied work in resource-aware scheduling, sensor networks, and high-performance computing that considers both the relative costs of communication and computation, often bundled together in an energy cost. However, these studies have been largely empirical rather than theoretical, and much of the work considers distributed algorithms (as opposed to our centralized setting).

\paragraph{AllReduce in HPC}
There is much related applied work  in the high-performance computing space on AllReduce~\cite{rabenseifner2004optimization, grama2003introduction}.
%AllReduce resembles our problem more closely than the more widely-known MapReduce because it requires some node in the network to compute a function of all nodes' inputs before dissemination to all nodes in the network, whereas MapReduce distributes the computation in parallel over a set of worker nodes with no guarantee that any node will eventually compute the function of all nodes' inputs \cite{dean2008mapreduce}.
However, while there has been significant research on communication-efficient AllReduce algorithms, there has been relatively little work that explicitly considers the cost of computation, and even less work that considers the construction of optimal topologies for efficient distributed computation. Researchers have empirically evaluated the performance of different models of communication \cite{oden2014energy, klenk2015analyzing} and have proven (trivial) lower bounds for communication without considering computation \cite{patarasuk2007bandwidth, patarasuk2009bandwidth}. 
%Additionally, there is work on developing adaptive AllReduce algorithms to, for instance, account for the presence of skew in the nodes of the cluster, but these algorithms likewise do not consider the cost of computation \cite{mamidala2004efficient}. 
Indeed, to the best of our knowledge, the extent to which they consider computation is through an additive penalty that consists of a multiplicative factor times the size of all inputs at all nodes, as in the work of Jain et al.~\cite{jain2012collectives}; crucially, this penalty is the same for any schedule and cannot be reduced via intelligent scheduling. Therefore, there do not seem to exist theoretical results for efficient algorithms that consider both the cost of communication and computation.

\paragraph{Resource-Aware Scheduling}
In the distributed computation space, people have considered resource-aware scheduling on a completely connected topology with different nodes having different loads. Although this problem considers computation-aware communication, these studies are much more empirical than theoretical, and only consider distributed solutions as opposed to centralized algorithms \cite{viswanathan2007resource, marchal2005realistic}.

\paragraph{Sensor Networks} Members of the sensor networks community have studied the problem of minimizing an energy cost, which succinctly combines the costs of communication and computation. However, sensor networks involve rapidly-changing, non-static topologies \cite{akyildiz2002survey,akyildiz2002wireless}, which means that their objective is not to construct a fixed, optimal topology, but rather to develop adaptive algorithms for minimizing total energy cost with respect to an objective function \cite{anastasi2009energy}.

\section{Deferred Figures}
\begin{figure}[H]
	\centering
	\begin{subfigure}{.49\textwidth}
		\includegraphics[width=\textwidth]{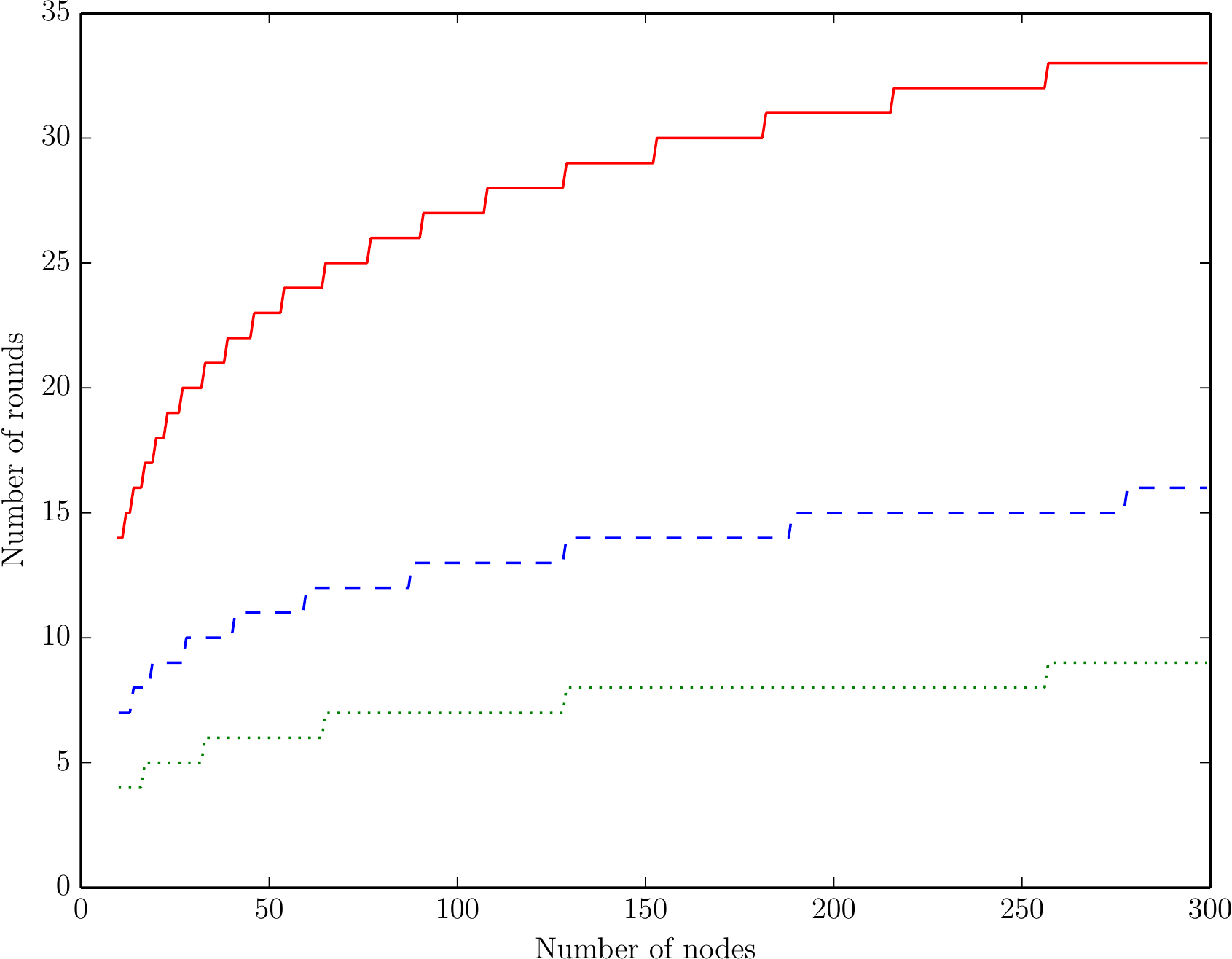}
		\caption{$t_c = 1$, $t_m = 2$}
	\end{subfigure}
	\begin{subfigure}{.49\textwidth}
		\includegraphics[width=\textwidth]{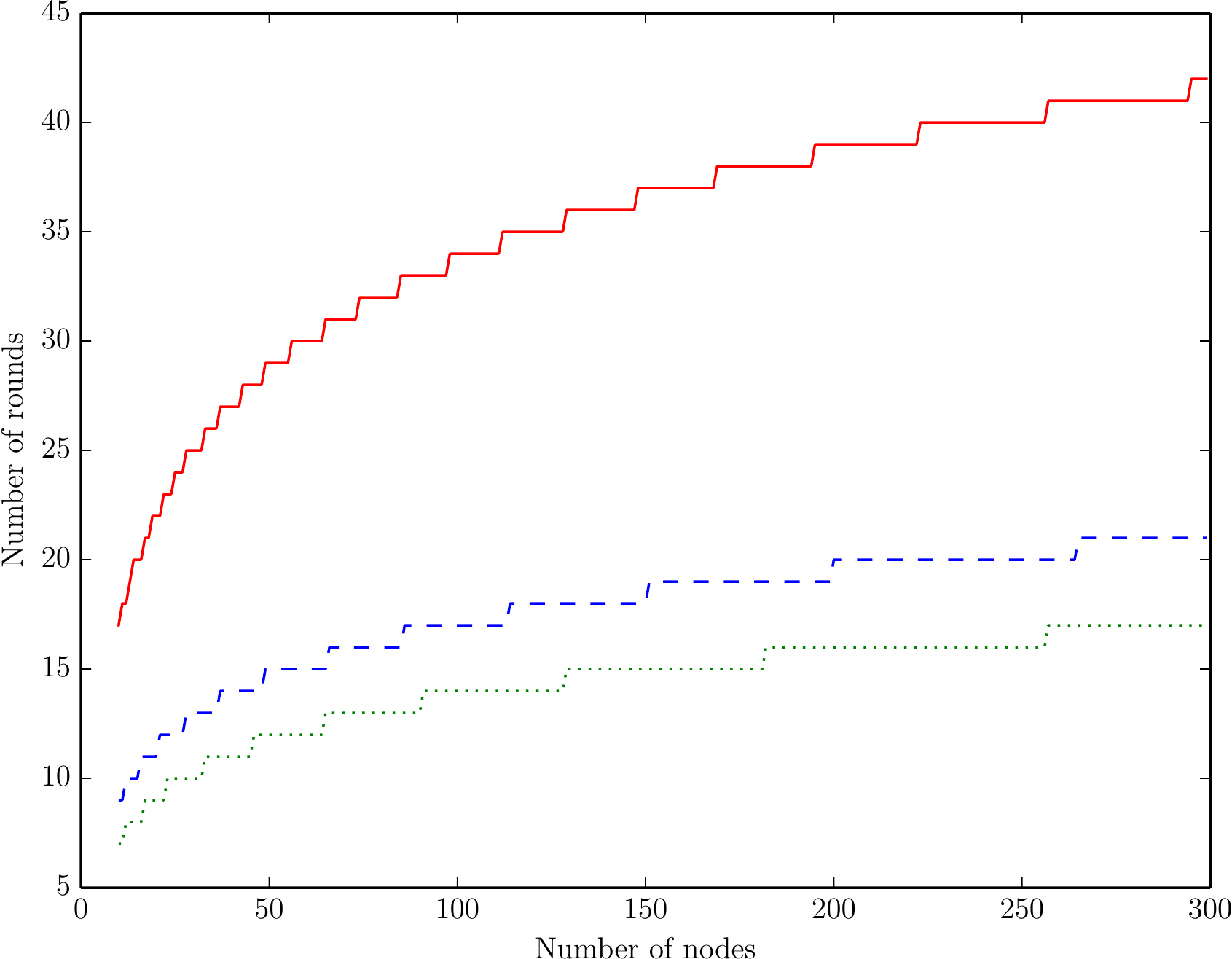}
		\caption{$t_c = 2$, $t_m = 1$}
	\end{subfigure}
	
	\caption{An illustration of the optimal schedule length for different sized trees. The solid red line is the number of rounds taken by greedy aggregation with pipelining on a binary tree (i.e., $\lceil 2 \cdot t_c \cdot \log n + t_m \cdot \log n \rceil$); the dashed blue line is the number of rounds taken by greedy aggregation on $T_n^*$; and the dotted green line is the trivial lower bound of $\lceil t_c \cdot \log n \rceil$ rounds. Note that though we illustrate the  trivial lower bound of $\lceil t_c \cdot \log n \rceil$ rounds, the true lower bound is given by the number of rounds taken by greedy aggregation on $T_n^*$.}
	\label{fig:optSchedLength}
\end{figure}

\section{Proof of Lemma~\ref{lem:opttm}}
\label{appsubsec:thmopt}

Using our combining and shortcutting insights, we establish the following structure on a schedule which solves \gump on $K_{N^*(R)}$ in $R$ rounds when $t_c < t_m$.

\begin{lemma}\label{lem:modSched}
	When $t_c < t_m$, for all $R \in \mathbb{Z}_0^+$, there exists a schedule, $\tilde{S}^*$, of length $R$ that solves \gump on $K_{N^*(R)}$ such that the terminus of $\tilde{S}^*$, $v_t$, never communicates and every computation performed by $v_t$ involves a token that contains $v_t$'s singleton \wad, $a_{v_t}$.
\end{lemma}

\begin{proof}
	Let $S^*$ be some arbitrary schedule of length $R$ which solves \gump on $K_{N^*(R)}$; we know that such a schedule exists by definition of $N^*(R)$. We first show how to modify $S^*$ into another schedule, $S_{1-4}^*$, which not only also solves \gump on $K_{N^*(R)}$ in $R$ rounds, but which also satisfies the following four properties.
	\begin{enumerate}[(1)]
		\item $v$ only sends at time $t$ if $v$ at time $t$ has exactly one \wad for $t \in [R]$;
		\item if $v$ sends in round $t$ then $v$ does not receive any \wads in rounds $[t, t + t_m]$ for $t \in [R]$; 
		\item if $v$ sends in round $t$ then $v$ is idle during rounds $t' > t$ for $t \in [R]$;
		\item the terminus never communicates.
	\end{enumerate}
	
	\paragraph{Achieving property (1).}
	Consider an optimal schedule $S^*$. We first show how to modify $S^*$ to an $R$-round schedule $S_1^*$ that solves \gump on $K_{N^*(R)}$ and satisfies property (1). We use our combining insight here. Suppose that (1) does not hold for $S^*$; i.e., a node $v$ sends a \wad $a_1$ to node $u$ at time $t$ and $v$ has at least one other \wad, say $a_2$, at time $t$. We modify $S^*$ as follows. At time $t$, node $v$ combines $a_1$ and $a_2$ into a \wad which it then performs operations on (i.e., computes and sends) as it does to $a_2$ in the original schedule. Moreover, node $u$ pretends that it receives \wad $a_1$ at time $t + t_m$: any round in which $S^*$ has $u$ compute on or communicate $a_1$, $u$ now simply does nothing; nodes that were meant to receive $a_1$ do the same. It is easy to see that by repeatedly applying the above procedure to every node when it sends when it has more than one \wad, we can reduce the number of \wads every node has whenever it sends to at most one. The total runtime of this schedule is no greater than that of $S^*$, namely $R$, because $t_c < t_m$. Moreover, it clearly still solves \gump on $K_{N^*(R)}$. Call the schedule $S_1^*$.
	
	\paragraph{Achieving properties (1) and (2).} Now, we show how to modify $S_1^*$ into $S_{1-2}^*$ such that properties (1) and (2) both hold. Again, $S_{1-2}^*$ is of length $R$ and solves \gump on $K_{N^*(R)}$. We use our shortcutting insight here.
	Suppose that (2) does not hold for $S_1^*$; i.e., there exists a $v$ that receives a \wad $a_1$ from node $u$ while sending another \wad $a_2$ to node $u'$. We say that node $u$ is \textit{bothering} node $v$ in round $t$ if node $u$ communicates a \wad $a_1$ to $v$ in round $t$, and node $v$ communicates a \wad $a_2$ to node $u' \in V \setminus\{v, u\}$ in round $[t, t + t_m]$. Say any such pair is a \emph{bothersome} pair. Furthermore, given a pair of nodes $(u,v)$ and round $t$ such that node $u$ is bothering node $v$ in round $t$, let the \textit{resolution} of $(u,v)$ in round $t$ be the modification in which $u$ sends its \wad directly to the node $u'$ to which $v$ sends its \wad.
	% For each $t$ and for every pair of nodes $(u,v)$ such that node $u$ is bothering node $v$ in round $t$, it is possible to modify $\tilde{S}^*$ into a schedule in which node $u$ is no longer bothering node $v$ by having node $u$ send its \wad directly to $u'$. 
	Note that each resolution does not increase the length of the optimal schedule because, by the definition of bothering, this will only serve as a shortcut; $u'$ will receive a \wad from $u$ at the latest in the same round it would have received a \wad from $v$ in the original schedule, and nodes $u'$ and $v$ can pretend that they received \wads from $v$ and $u$, respectively. However, it may now be the case that node $u$ ends up bothering node $u'$. We now show how to repeatedly apply resolutions to modify $S_1^*$ into a schedule $S_{1-2}^*$ in which no node bothers another in any round $t$.
	
	Consider the graph $B_t(S_1^*)$ where the vertices are the nodes in $G$ and there exists a directed edge $(u,v)$ if node $u$ is bothering node $v$ in round $t$ in schedule $S_1^*$. First, consider cycles in $B_t(S_1^*)$. Note that, for any time $t$ in which $B_t(S_1^*)$ has a cycle, we can create a schedule $\tilde{S}_1^*$ in which no nodes in any cycle in $B_t(S_1^*)$ send their \wads in round $t$; rather, they remain idle this round and pretend they received the \wad they would have received under $S_1^*$. Clearly, this does not increase the length of the optimal schedule and removes all cycles in round $t$. Furthermore, this does not violate property (1) because fewer nodes send \wads in round $t$, and no new nodes send \wads in round $t$. 
	
	Therefore, it suffices to consider an acyclic, directed graph $B_t(\tilde{S}_1)$. Now, for each round $t$, we repeatedly apply resolutions until no node bothers any other node during that round. Note that for every $t$, each node can only be bothering at most one other node because nodes can only send one message at a time. This fact, coupled with the fact that $B_t(\tilde{S}_1)$ is acyclic, means that $B_t(\tilde{S}_1)$ is a DAG where nodes have out-degree 1. It is not hard to see that repeatedly applying resolutions to a node $v$ which bothers another node will decrease the number of edges in $B_t(\tilde{S}_1)$ by 1. Furthermore, because there are $n$ total nodes in the network, the number of resolutions needed for any node $v$ at time $t$ is at most $n$.
	%\ap{I'm not sure what it means to apply $n$ resolutions to a single node.}
	
	Furthermore, repeatedly applying resolutions to $B_t(\tilde{S}_1)$ for times $t=1,\ldots,R$ in order results in a schedule $S_{1-2}^*$ with no bothersome pairs at any time $t$ and that still satisfies property (1), and so schedule $S_{1-2}^*$ satisfies properties (1) and (2). Since each resolution did not increase the length of the schedule we also have that $S_{1-2}^*$ is of length $R$. Lastly, $S_{1-2}^*$ clearly still solves \gump on $K_{N^*(R)}$.
	
	\paragraph{Achieving properties (1) - (3).} Now, we show how to modify $S_{1-2}^*$ into $S_{1-3}^*$ which satisfies properties (1), (2), and (3). We use our shortcutting insight here as well as some new ideas. Given $S_{1-2}^*$, we show by induction over $k$ from $0$ to $R - t_m$, where $R$ is the length of an optimal schedule, that we can modify $S_{1-2}^*$ such that if a node finishes communicating in round $R-k$ (i.e., begins communicating in round $R-k-t_m$), it remains idle in rounds $t' \in (R - k, R]$ in the modified optimal schedule. The base case of $k = 0$ is trivial: If a node communicates in round $R - t_m$, it must remain idle in round $R$ because the entire schedule is of length $R$.
	
	Suppose there exists a node $v$ that finishes communicating in round $t = R - k$ but is not idle in some round $t' > R - k$ in $S_{1-2}^*$; furthermore, let round $t'$ be the first round after $t$ in which node $v$ is not idle. By property (1), node $v$ must have sent its only \wad away in round $t$, and therefore node $v$ must have received at least one other \wad after round $t$ but before round $t'$. We now case on the type of action node $v$ performs in round $t'$. 
	
	\begin{itemize}
		\item If node $v$ communicates in round $t'$, it must send a \wad it received after time $t$ but before round $t'$. Furthermore, as this is the first round after $t$ in which $v$ is not idle, $v$ cannot have performed any computation on this \wad, and by the inductive hypothesis, $v$ must remain idle from round $t' + t_m$ on. Therefore, $v$ receives a \wad $a_u$ from some node $u$ and then forwards this \wad to node $u'$ at time $t'$. One can modify this schedule such that $u$ sends $a_u$ directly to $u'$ instead of sending to $v$.
		
		\item If node $v$ computes in round $t'$, consider the actions of node $v$ after round $t' + t_c$. Either $v$ eventually performs a communication after some number of computations, after which point it is idle by the inductive hypothesis, or $v$ only ever performs computations from time $t'$ on.
		
		In round $t'$, $v$ must combine two \wads it received after time $t + t_m$ by property (1). Note that two distinct nodes must have sent the two \wads to $v$ because, by the inductive hypothesis, each node that sends after round $t$ remains idle for the remainder of the schedule. Therefore, the nodes $u_1$ and $u_2$ that sent the two \wads to $v$ must have been active at times $t_1', t_2' > t$, where $t_1 \leq t_2$, after which they remain idle for the rest of the schedule. Call the tuple $(v, u_1, u_2)$ a \textit{switchable} triple. We can modify the schedule to make $v$ idle at round $t'$ by picking the node that first sent to $v$ and treating it as $v$ while the original $v$ stays idle for the remainder of the schedule. In particular, we can modify $S_{1-2}^*$ such that, without loss of generality, $u_2$ sends its \wad to $u_1$ and $u_1$ performs the computation that $v$ originally performed in $S_{1-2}^*$. Note that this now ensures that $v$ will be idle in round $t'$ and does not increase the length of the schedule, as $u_1$ takes on the role of $v$. Furthermore, node $u_1$'s new actions do not violate the inductive hypothesis: Either $u_1$ only ever performs computations after time $t'$, or it eventually communicates and thereafter remains idle.
		
		We can repeat this process for all nodes that are not idle after performing a communication in order to produce a schedule $S_{1-3}^*$ in which property (3) is satisfied. 
	\end{itemize}
	
	First, notice that these modifications do not change the length of $S_{1-2}^*$: in the first case $u'$ can still pretend that it receives $a_u$ at time $t' + t_m$ even though it now receives it in an earlier round and in the second case $u_2$ takes on the role of $v$ at the expense of no additional round overhead. Also, it is easy to see that $S_{1-3}^*$ still solves \gump on $K_{N^*(R)}$.
	
	We now argue that the above modifications preserve (1) and (2). First, notice that the modifications we do for the first case do not change when any nodes send and so (1) is satisfied. In the second case, because we switch the roles of nodes, we may potentially add a send for a node. However, note that we only require a node $u_1$ to perform an additional send when it is part of a switchable triple $(v, u_1, u_2)$, and $u_1$ takes on the role of $v$ in the original schedule from time $t'$ on. However, because $S_{1-2}^*$ satisfies (1), $u$ was about to send its only \wad away and therefore only had one \wad upon receipt of the \wad from $u_2$. Therefore, because $u_1$ performs the actions that $v$ performs in $S_{1-2}^*$ from time $t'$ on, and because at time $t'$, both $u_1$ and $v$ have exactly two \wads, (1) is still satisfied by $S_{1-3}^*$. Next, we argue that (3) is a strictly stronger condition than (2). In particular, we show that since $S_{1-3}^*$ satisfies (3) it also satisfies (2). Suppose for the sake of contradiction that $S_{1-3}^*$ satisfies (3) but not (2). Since (2) is not satisfied there must exist some node $v$ that sends in some round $t$ to, say node $u$, but receives a \wad in some round in $[t, t+t_m]$. By (3) it then follows that $v$ is idle in all rounds after $t$. %Since $v$ is idle in all rounds after $t$ and receives a \wad in some round in $[t, t+t_m]$, $v$ must be the terminus. 
	However, $u$ also receives a \wad in round $t + t_m$. Therefore, in round $t+t_m$, two distinct nodes have \wads, one of which is idle in all rounds after $t+t_m$; this contradicts the fact that $S_{1-3}^*$ solves \gump. Thus, $S_{1-3}^*$ must also satisfy (2)
	
	\paragraph{Achieving properties (1) - (4).}
	It is straightforward to see that $S_{1-3}^*$ also satisfies property (4). Indeed, by property (3), if the terminus ever sends in round $t < R - t_c$, then the terminus must remain idle during rounds $t' > t$, meaning it must be idle in round $R - t_c$ which contradicts the fact that in this round the terminus performs a computation. Therefore, $S_{1-4}^* = S_{1-3}^*$ satisfies properties (1) - (4), and we know that there exists an optimal schedule in which $v_t$ is always either computing or idle. 
	
	\paragraph{Achieving the final property.}
	We now argue that we can modify $S_{1-4}^*$ into another optimal schedule $\tilde{S}^*$ such that every computation done at the terminus $v_t$ involves a \wad that contains the original singleton \wad that started at the terminus. Suppose that in $S_{1-4}^*$, $v_t$ performs computation that does not involve $a_{v_t}$. Take the first instance in which $v_t$ combines \wads $a_1$ and $a_2$, neither of which contains $a_{v_t}$, in round $t$. Because this is the first computation that does not involve a \wad containing $a_{v_t}$, both $a_1$ and $a_2$ must have been communicated to the terminus in round $t - t_m$ at the latest.
	
	Consider the earliest time $t' > t$ in which $v_t$ computes a \wad $a_{comb}$ that contains all of $a_1$, $a_2$, and $a_{v_t}$. We now show how to modify $S_{1-4}^*$ into $\tilde{S}'$ such that $v_t$ computes a \wad $a_{comb}'$ at time $t'$ that contains all of $a_1$, $a_2$, and $a_{v_t}$ and is at least the size of $a_{comb}$ by having nodes swap roles in the schedule between times $t$ and $t'$. Furthermore, because the rest of the schedule remains the same after time $t'$, this implies that $\tilde{S}'$ uses at most as many rounds as $S_{1-4}^*$, and therefore that $\tilde{S}'$ uses at most $R$ rounds. 
	
	The modification is as follows. At time $t$, instead of having $v_t$ combine \wads $a_1$ and $a_2$, have $v_t$ combine one of them (without loss of generality, $a_1$) with the \wad containing $a_{v_t}$. Now, continue executing $S_{1-4}^*$ but substitute $a_2$ for the \wad containing $a_{v_t}$ from round $t$ on; this is a valid substitution because $v_t$ possesses $a_2$ at time $t$. In round $t'$, $v_t$ computes a \wad $a_{comb}' = a_{comb}$; the difference from the previous schedule is that the new schedule has one fewer violation of property (4), i.e., one fewer round in which it computes on two \wads, neither of which contains $a_{v_t}$.
	
	We repeat this process for every step in which the terminus does not compute on the \wad containing $a_{v_t}$, resulting in a schedule $\tilde{S}^*$ in which the terminus is always combining a communicated \wad with a \wad containing its own singleton \wad. Note that these modifications do not affect properties (1) - (4) because this does not affect the sending actions of any node, and therefore $\tilde{S}^*$ still satisfies properties (1) - (4). It easily follows, then, that $\tilde{S}^*$ solves \gump on $K_{N^*(R)}$ in $R$ rounds. Thus, $\tilde{S}^*$ is a schedule of length $R$ that solves \gump on $K_{N^*(R)}$ in which every computation the terminus $v_t$ does is on two \wads, one of which contains $a_{v_t}$, and, by (4), the terminus $v_t$ never communicates.
\end{proof}

Having shown that the schedule corresponding to $N^*(R)$ can be modified to satisfy a nice structure when $t_c < t_m$, we can conclude our recursive bound on $N^*(R)$.
%\begin{lemma}
%	\label{lem:opttm}
%	When $t_c < t_m$, for $R \in \mathbb{Z}_0^+$ it holds that $N^*(R) = N^*(R - t_c) + N^*(R - t_c - t_m)$ for $R \geq t_c + t_m$.
%\end{lemma}

\opttm*
\begin{proof}%[Proof of \Cref{lem:opttm}]
	Suppose $R \geq t_c + t_m$. We begin by applying \Cref{lem:modSched} to show that $N^*(R) \leq N^*(R - t_c) + N^*(R - t_c - t_m)$. Let $v_t$ be the terminus of the schedule $\tilde{S}^*$ using $R$ rounds as given in \Cref{lem:modSched}. By \Cref{lem:modSched}, in all rounds after round $t_m$ of $\tilde{S}^*$ it holds that $v_t$ is either computing on a \wad that contains $a_{v_t}$ or busy because it did such a computation. Notice that it follows that every \wad produced by a computation at $v_t$ contains $a_{v_t}$.
	
	Now consider the last \wad produced by our schedule. Call this \wad $a$. By definition of the terminus, $a$ must be produced by a computation performed by $v_t$, combining two \wads, say $a_1$ and $a_2$, in round $R-t_c$ at the latest. Since every computation that $v_t$ does combines two \wads, one of which contains $a_{v_t}$, without loss of generality let $a_1$ contain $a_{v_t}$. 
	
	We now bound the size of $a_1$ and $a_2$. Since $a_1$ exists in round $R-t_c$ we know that it is of size at most $N^*(R - t_c)$. Now consider $a_2$. Since every \wad produced by a computation at $v_t$ contains $a_{v_t}$ and $a_2$ does not contain $a_{v_t}$ it follows that $a_2$ must either be a singleton \wad that originates at a node other than $v$, or $a_2$ was produced by a computation at another node. Either way, $a_2$ must have been sent to $v$, who then performed a computation on $a_2$ in round $R-t_c$ at the latest. It follows that $a_2$ exists in round $R - t_c - t_m$, and so $a_2$ is of size no more than $N^*(R - t_c - t_m)$.
	
	Since the size of $a$ just is the size of $a_1$ plus the size of $a_2$, we conclude that $a$ is of size no more than $N^*(R - t_c) + N^*(R - t_c - t_m)$. Since, $\tilde{S}^*$ solves \gump on a complete graph of size $N^*(R)$, we have that $a$ is of size $N^*(R)$ and so we conclude that $N^*(R) \leq N^*(R - t_c) + N^*(R - t_c - t_m)$ for $R \geq t_c + t_m$ when $t_c < t_m$.
	
	Lastly, since $N^*(R) \geq N^*(R - t_c) + N^*(R - t_c - t_m)$ for $R \geq t_c + t_m$ by \Cref{lem:NStarLB}, we conclude that $N^*(R) = N^*(R - t_c) + N^*(R - t_c - t_m)$ for $R \geq t_c + t_m$ when $t_c < t_m$.
\end{proof}

\section{Proof of Theorem~\ref{thm:hardapx}}
\label{app:hard}

As a warmup for our hardness of approximation result, and to introduce some of the techniques, we begin with a proof that the decision version of \gump is NP-complete in \Cref{app:hardness}. We then prove the hardness of approximation result in \Cref{app:hardapx}.

\subsection{NP-Completeness (Warmup)}
\label{app:hardness}
An instance of the decision version of \gump is given by an instance of \gump and a candidate $\ell$. An algorithm must decide if there exists a schedule that solves \gump in at most $\ell$ rounds.

We reduce from $k$-dominating set. 
\begin{definition}[$k$-dominating set]
	An instance of $k$-dominating set consists of a graph $G = (V,E)$; the decision problem is to decide whether there exists $\kappa \subseteq V$ where $|\kappa| = k$ such that for all $v \in V \setminus \kappa$ there exists $\nu \in \kappa$ such that $(v, \nu) \in E$. 
\end{definition}

Recall that $k$-dominating set is NP-complete.

\begin{lemma}[Garey and Johnson \cite{garey2002computers}]
	$k$-dominating set is NP-complete.
\end{lemma}

Given an instance of $k$-dominating set, we would like to transform $G$ into another graph $G'$ in polynomial time such that $G$ has a $k$-dominating set iff there exists a \gump schedule of some particular length for $G'$ for some values of $t_c$ and $t_m$. 

We begin by describing the intuition behind the transformation we use, which we call $\Psi$. Any schedule on graph $G$ in which every node only performs a single communication and which aggregates all \wads down to at most $k$ \wads corresponds to a $k$-dominating set of $G$; in particular, those nodes that do computation form a $k$-dominating set of $G$. If we had a schedule of length $< 2t_m$ which aggregated all \wads down to $k$ \wads, then we could recover a $k$-dominating set from our schedule. However, our problem aggregates down to only a single \wad, not $k$ \wads. Our crucial insight, here, is that by structuring our graph such that a single node, $a$, must perform a great deal of computation, $a$ must be the terminus of any short schedule. The fact that $a$ must be the terminus and do a great deal of computation, in turn, forces any short schedule to aggregate all \wads in $G$ down to at most $k$ \wads at some point, giving us a $k$-dominating set.

Formally, $\Psi$ is as follows. $\Psi$ takes as input a graph $G$ and a value for $t_m$ and outputs $G' = (V', E')$. $G'$ has $G$ as a sub-graph and in addition has auxiliary node $a$ where $a$ is connected to all $v \in V$; $a$ is also connected to dangling nodes $d \in \beta$, where $|\beta| = \Delta + t_m$, along with a special dangling node $d^*$.\footnote{$\Delta$ is the max degree of $G$.} Thus, $G' = (V \cup \{a, d^*\} \cup \beta, E \cup \{(a,v') : v' \in V' \setminus \{a\}\})$. See \Cref{fig:R}.

\begin{figure}
	\centering
	\begin{subfigure}[t]{0.3\textwidth}
		\centering
		\includegraphics[scale=0.13,page=1]{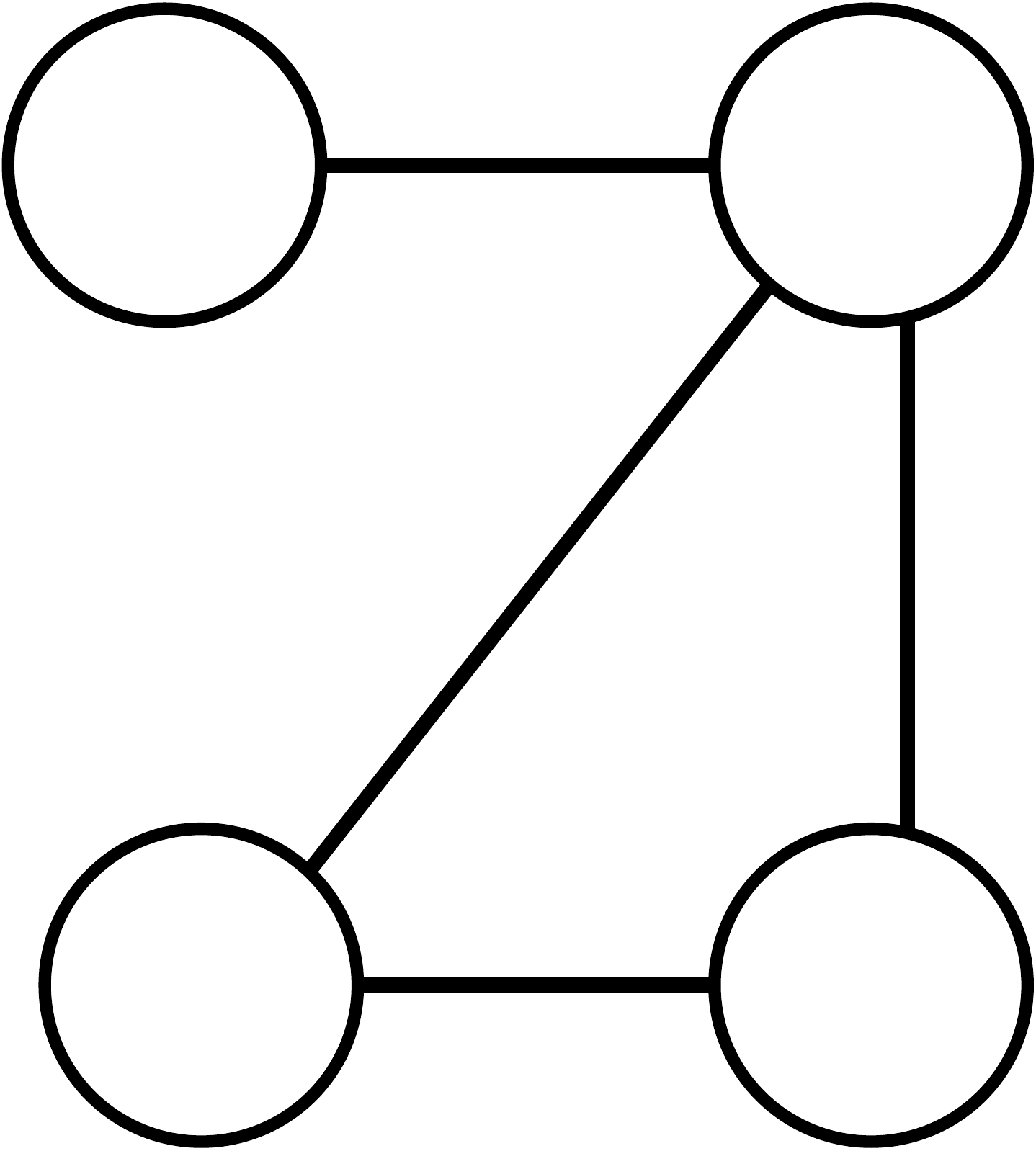}
		\caption{$G$}
	\end{subfigure}%
	~
	\begin{subfigure}[t]{0.3\textwidth}
		\centering
		\includegraphics[scale=0.13,page=2]{figs/reduction.pdf}
		\caption{$\Psi(G,1)$}
	\end{subfigure}%
	\caption{An example of $\Psi$ for a given graph $G$ and $t_m = 1$. Nodes and edges added by $\Psi$ are dashed and in blue. Notice that $|\beta| = \Delta + t_m = 3 + 1 = 4$.}
	\label{fig:R}
\end{figure}

We now prove that the optimal \gump schedule on $G' = \Psi(G, t_m)$ can be upper bounded as a function of the size of the minimum dominating set of $G$.
\begin{lemma}\label{lem:optGumUpperBound}
	The optimal \gump schedule on $G' = \Psi(G, t_m)$ is of length at most $2t_m + \Delta + k^*$ for $t_c = 1$, where $k^*$ is the minimum dominating set of $G$.
\end{lemma}
\begin{proof}
	We know by definition of $k^*$ that there is a dominating set of size $k^*$ on $G$. Call this set $\kappa$ and let $\sigma : V \rightarrow \kappa$ map any given $v\in V$ to a unique node in $\kappa$ that dominates it. We argue that it must be the case that \gump requires at most $2t_m + t_c(\Delta + k^*)$ rounds on $G'$ for $t_c = 1$. Roughly, we solve \gump by first aggregating at $\kappa$ and then aggregating at $a$.
	
	In more detail, in stage $1$ of the schedule, every $d \in \beta$ sends to $a$, every node $v \in V$ sends to $\sigma(v)$ and $a$ sends to $d^*$. This takes $t_m$ rounds. In stage $2$, each node does the following in parallel. Node $d^*$ computes and sends its single \wad to $a$. Each $\nu \in \kappa$ computes until it has a single \wad and sends the result to $a$. Node $a$ combines all \wads from $\beta \cup \{d^*\}$. Node $d^*$ takes $1 + t_m$ rounds to do this. Each $\nu \in \kappa$ takes at most $\Delta + t_m$ rounds to do this. Node $a$ takes $\Delta+t_m$ rounds to do this since $a$ will receive $d^*$'s \wad after $t_m + 1$ rounds (and $\Delta \geq 1$ without loss of generality). Thus, stage $2$, when done in parallel, takes $\Delta + t_m$ rounds. At this point $a$ has $k^*+1$ \wads and no other node in $G'$ has a \wad. In stage $3$, $a$ computes until it has only a single \wad, which takes $k^*$ rounds.
	
	In total the number of rounds used by this schedule is $t_m + \Delta + t_m + k^* = 2t_m + \Delta + k^*$. Thus, the total number of rounds used by the optimal \gump schedule on $G'$ is at most $2t_m + \Delta + k^*$.
\end{proof}

Next, we show that any valid \gump schedule on $G' = \Psi(G, t_m)$ that has at most two serialized sends corresponds to a dominating set of size bounded by the length of the schedule.

\begin{lemma}\label{lem:kapGivesDS}
	Given $G' = \Psi(G, t_m)$ and a \gump schedule $S$ for $G'$ where $|S| < 3t_m$, $t_c = 1$, $\kappa = \{ v : v \in G, \text{ $v$ sends to $a$ in $S$}\}$ is a dominating set of $G$ of size $|S| - 2t_m - \Delta$.
\end{lemma}
\begin{proof}
	Roughly, we argue that $a$ must be the terminus of $S$ and must perform at most $|S|-2t_m - \Delta$ computations on \wads from $G$, each of which is the aggregation of a node's \wad and some of its neighbors' \wads. We begin by arguing that $a$ must be the terminus.
	
	First, we prove that no $d \in \beta \cup \{d^*\}$ is the terminus of $S$. Suppose for the sake of contradiction that some $\bar{d} \in \beta \cup \{d^*\}$ is the terminus. Since our schedule takes fewer than $3t_m$ rounds, we know that every node sends a \wad that is not just the singleton \wad with which it starts at most once. Thus, $a$ sends \wads that are not just the singleton \wad that it starts with at most once. Since $|\beta \cup \{d^*\} \setminus \{\bar{d}\}| = \Delta  + t_m$ and $a$ is the only node connected to these nodes, we know that every singleton \wad that originates in $\beta \cup \{d^*\} \setminus \{\bar{d}\}$ must travel through $a$. Moreover, since $a$ sends \wads that are not just the singleton \wad that it starts with at most once, $a$ must send all such \wads as a single \wad. It follows that $a$ must perform at least $\Delta  + t_m$ computations, but then our entire schedule takes at least $t_m + \Delta  + t_m + t_m = 3t_m + \Delta > 3t_m$ rounds\,---\,a contradiction to our assumption that our schedule takes less than $3 t_m$ rounds. 
	
	We now argue that no $v \in G$ is the terminus. Suppose for the sake of contradiction that some $\bar{v} \in V$ is the terminus. Again, we know that $a$ sends \wads that are not just the singleton \wad that it starts with at most once. Thus, every \wad in $\beta \cup \{d^*\}$ must travel through $a$, meaning that $a$ must perform $\Delta  + t_m + 1$ computations. It follows that the schedule takes $t_m + \Delta  + t_m + t_m + 1> 3t_m$ rounds, a contradiction to our assumption that the schedule takes $< 3 t_m$ rounds.
	
	Thus, since no $d \in \beta \cup \{d^*\}$ and no $v \in G$ is the terminus, we know that $a$ must be the terminus.
	
	We now argue that $a$ sends a \wad in the first round and this is the only time that $a$ sends (i.e., the only thing that $a$ sends is the singleton \wad that it starts with, which it sends immediately). Assume for the sake of contradiction that $a$ sends a \wad that it did not start with. It must have taken at least $t_m$ rounds for this \wad to arrive at $a$ and at least an additional $t_m$ rounds for $a$ to send a \wad containing it. Moreover, since $a$ is the terminus, a \wad containing this \wad must eventually return to $a$ and so an additional $t_m$ rounds are required. Thus, at least $3t_m$ rounds are required if $a$ sends a \wad other than that with which it starts, a contradiction to the fact that our schedule takes $< 3t_m$ rounds.
	
	Thus, since $a$ is the terminus, our schedule solves \gump in fewer than $3t_m$ rounds, and no computations occur in the first $t_m$ rounds, $a$ does at most $|S| - t_m$ computations. Since $a$ never sends any \wad aside from its singleton \wad, and $a$ is the only node to which $\beta \cup \{d^*\}$ are connected, we know that $a$ must combine all \wads of nodes in $\beta \cup \{d^*\}$, where $a$ must take $\Delta + t_m$ rounds to do so. Thus, since $a$ takes $\Delta + t_m$ rounds to aggregate \wads from $\beta \cup \{d^*\}$ and it performs at most $|S| - t_m$ computations in total, $a$ must receive at most $|S| - 2t_m - \Delta$ \wads from $G$. It follows that $|\kappa| \leq |S| - 2t_m - \Delta$.
	
	Since each \wad sent by a node in $\kappa$ to $a$ must be sent at the latest in round $|S|-t_m$ and since $|S| < 3t_m$, we have that every \wad sent by a node in $\kappa$ is formed in fewer than $2t_m$ rounds. It follows that each such \wad is formed by \wads that travel at most 1 hop in $G$. Since every node in $G$ must eventually aggregate its \wads at $a$, it follows that every node in $G$ is adjacent to a node in $\kappa$. Thus $\kappa$ is a dominating set of $G$, and as shown before $|\kappa| \leq |S| - 2t_m - \Delta$.
\end{proof}

Having shown that the optimal \gump schedule of $G' = \Psi(G, t_m)$ is closely related to the size of the minimum dominating set, we prove that \gump is NP-complete.
\begin{theorem}
	The decision version of \gump is NP-complete.
\end{theorem}
\begin{proof}
	The problem is clearly in NP. To show hardness, we reduce from $k$-dominating set. Specifically, we give a polynomial-time Karp reduction from $k$-dominating set to the decision version of \gump. 
	
	Our reduction is as follows. First, run $\Psi(G, t_m)$ for $t_m = \Delta + k + 1$ to get back $G'$. Next, return a decision version instance of \gump given by graph $G'$ with $t_m = \Delta + k + 1$, $t_c = 1$ and $\ell=2t_m + \Delta + k$. We now argue that $G'$ has a schedule of length $\ell$ iff $G$ has a $k$-dominating set.
	\begin{itemize}
		\item Suppose that $G$ has a $k$-dominating set. We know that $k \geq k^*$, where $k^*$ is the minimum dominating of $G$, and so by \Cref{lem:optGumUpperBound} we know that $G'$ has a schedule of length at most $2t_m + \Delta + k^* \leq 2t_m + \Delta + k$.
		\item Suppose that $G'$ has a \gump schedule $S$ of length at most $2t_m + \Delta + k$. Notice that by our choice of $t_m$, we have that $|S| = 2t_m + \Delta + k < 3t_m$ and so by \Cref{lem:kapGivesDS} we know that $\kappa = \{ v : v \in G, \text{ $v$ sends to $a$ in $S$}\}$ is a dominating set of $G$ of size $|S| - 2t_m - \Delta$. Since $|S| \leq 2t_m + \Delta + k$ we conclude that $|\kappa|  = |S| - 2t_m - \Delta\leq k$.
	\end{itemize}
	
	Lastly, notice that our reduction, $\Psi$, runs in polynomial time since it adds at most a polynomial number of vertices and edges to $G$. Thus, we conclude that $k$-dominating is polynomial-time reducible to the decision version of \gump, and therefore the decision version of \gump is NP-complete.
\end{proof}

\subsection{Hardness of Approximation}
\label{app:hardapx}
We now show that unless $\text{P} = \text{NP}$ there exists no polynomial-time algorithm that approximates \gump multiplicatively better than $1.5$.

Recall that $k$-dominating set is $\Omega(\log n)$ hard to approximate.

\begin{lemma}[Dinur and Steurer \cite{dinur2014analytical}]\label{lem:MDSHardnessAPX}
	Unless $\text{P} = \text{NP}$ every polynomial-time algorithm approximates minimum dominating set at best within a $(1-o(1))(\log n)$ multiplicative factor.
\end{lemma}

We prove hardness of approximation by using a $(1.5-\eps)$ algorithm for \gump to approximate minimum dominating set with a polynomial-time algorithm better than $O(\log n)$. Similar to our proof of NP-completeness, given input graph $G$ whose minimum dominating set we would like to approximate, we would like to transform $G$ into another graph $G'$ such that a $(1.5 - \eps)$-approximate \gump schedule for $G'$ allows us to recover an approximately minimum dominating set. 

One may hope to simply apply the transformation $\Psi$ from the preceding section to do so. However, it is not hard to see that the approximation factor on the minimum dominating set recovered in this way has dependence on $\Delta$, the maximum degree of $G$. If $\Delta$ is significantly larger than the minimum dominating set of $G$, we cannot hope that this will yield a good approximation to minimum dominating set. For this reason, before applying $\Psi$ to $G$, we duplicate $G$ a total of $\Delta/\eps$ times to create graph $G_\alpha$; this keeps $\Delta$ unchanged but increases the size of the minimum dominating set.\footnote{Since the max degree of $G$ and $G_\alpha$ are the same, throughout this section $\Delta$ will be used to refer to both the max degree of $G$ and the max degree of $G_\alpha$.} By applying $\Psi$ to $G_\alpha$ instead of $G$ to get back $G_\alpha'$ we are able to free our approximation factor from a dependence on $\Delta$. Lastly, we show that we can efficiently recover an approximate minimum dominating set for $G$ from an approximate \gump schedule for $G_\alpha'$ using our polynomial-time algorithm \textsc{DSFromSchedule}. Our full algorithm is given by \textsc{MDSApx}.

We first describe the algorithm\,---\,\textsc{DSFromSchedule}\,---\,we use to recover a minimum dominating set for $G$ given a \gump schedule for $G_\alpha' = \Psi(G_\alpha, t_m)$. We denote copy $i$ of $G$ as $G_i$. 

\begin{algorithm}
	\caption{\textsc{DSFromSchedule}}
	\label{alg:DSFromSchedule}
	\begin{algorithmic}
		\Statex \textbf{Input:} $G_\alpha' = \Psi(G_\alpha, t_m)$; a valid \gump schedule for $G_\alpha'$, $S$, of length $<3t_m$; $\epsilon$
		\Statex \textbf{Output:} A dominating set for $G$ of size $|S| - 2t_m - \Delta$
		\State $\mcK \gets \emptyset$
		\For{$i \in \left[\frac{\Delta}{\eps}\right]$}
		\State $\kappa_i \gets \{v \in V_i : \text{$v \in G_\alpha$ sends to $a$ in $S$}\}$
		\State $\mcK \gets \mcK \cup \{\kappa_i\}$
		\EndFor
		\State \Return $\argmin_{\kappa_i \in \mcK} |\kappa_i|$
	\end{algorithmic}
\end{algorithm}

\begin{lemma}\label{lem:DSFromSched}
	Given $G_\alpha'= \Psi(G_\alpha, t_m)$ and a valid \gump schedule $S$ for $G_\alpha'$ where $|S| < 3t_m$, $t_c = 1$ and $\eps \in (0, 1]$, \textsc{DSFromSchedule} outputs in polynomial time a dominating set of $G$ of size $\frac{\eps}{\Delta}\left(|S| - 2t_m - \Delta\right)$.\footnote{Since this lemma allows for $\eps \in (0, 1]$, it may appear that we will be able to achieve an arbitrarily good approximation for minimum dominating set. In fact, it might even seems as though we can produce a dominating set of size smaller than the minimum dominating set by simply letting $\eps$ be arbitrarily small. However, this is not the case. Intuitively, the smaller $\eps$ is, the larger $G_\alpha$ is and so the longer any feasible schedule $S$ must be. Thus, decreases in $\eps$ are balanced out by increases in $|S|$ with respect to the size of our dominating set, $\frac{\eps}{\Delta}\left(|S| - 2t_m - \Delta\right)$.}
\end{lemma}
\begin{proof}
	Polynomial runtime is trivial, so we focus on the size guarantee. By \Cref{lem:kapGivesDS} we know that $\kappa = \{ v : v \in G_\alpha, \text{ $v$ sends to $a$ in $S$}\}$ is a dominating set of $G_\alpha$ of size $|S| - 2t_m - \Delta$. Moreover, notice that $\kappa_i = \kappa \cap G_i$, and so it follows that $\kappa_i$ is a dominating set of $G_i$, or equivalently $G$ because $G_i$ is just a copy of $G$. Thus we have that $\argmin_{\kappa_i \in \mcK} |\kappa_i|$ will return a dominating set of $G$.
	
	We now prove that $\argmin_{\kappa_i \in \mcK} |\kappa_i|$ is small. Since each $\kappa_i$ is disjoint we have $\sum_{i=1}^{\Delta/\eps} |\kappa_i| = |\kappa| \leq |S| - 2t_m - \Delta$. Thus, by an averaging argument we have that there must be some $\kappa_i$ such that $|\kappa_i| \leq \frac{\eps}{\Delta}\left(|S| - 2t_m - \Delta\right)$. It follows that $\min_{\kappa_i \in \mcK} |\kappa_i| \leq \frac{\eps}{\Delta}\left(|S| - 2t_m - \Delta\right)$, meaning the $\kappa_i$ that our algorithm returns is not only a dominating set of $G$ but of size at most $\frac{\eps}{\Delta}\left(|S| - 2t_m - \Delta\right)$.
\end{proof}

Lastly, we combine $\Psi$ with \textsc{DSFromSchedule} to get \textsc{MDSApx}, our algorithm for approximating minimum dominating set. Roughly, \textsc{MDSApx} constructs $G_\alpha'$ by applying $\Psi$ to $G_\alpha$, uses a $(1.5 - \eps)$ approximation to \gump to get a schedule to $G_\alpha'$ and then uses \textsc{DSFromSchedule} to extract a minimum dominating set for $G$ from this schedule. \textsc{MDSApx} will carefully choose a $t_m$ that is large enough so that the schedule produced by the $(1.5 - \eps)$ approximation for \gump is of length $< 3t_m$ but also small enough so that the produced schedule  can be used to recover a small dominating set.

\begin{algorithm}
	\caption{\textsc{MDSApx}}
	\label{alg:apxMDS}
	\begin{algorithmic}
		\Statex \textbf{Input:} Graph $G$; $(1.5-\eps)$ \gump approximation algorithm $\mcA$
		\Statex \textbf{Output:} An $O(1/\eps)$-approximation for the minimum dominating set of $G$
		\State $\mcD \gets \emptyset$
		\For{$\hat{k} \in [n]$}
		\State $G_\alpha \gets \bigcup_{i=1}^{\Delta/\eps} G_i$
		\State $t_m \gets \frac{1}{\eps}\left(\Delta + \frac{\hat{k}\Delta}{\eps} \right) + 1$; $t_c \gets 1$ 
		\State $G_\alpha' \gets \Psi \left(G_\alpha, t_m \right)$
		\State $S_{\hat{k}} \gets \mcA\left(G_\alpha', \frac{\hat{k}}{\eps}, t_m, t_c \right)$
		\If{$|S_{\hat{k}}| < 3 t_m$}
		\State $\kappa_{\hat{k}} \gets \textsc{DSFromSchedule}(G_\alpha, S, \eps)$
		\State $\mcD \gets \mcD \cup \{\kappa_{\hat{k}}\} $ 
		\EndIf
		\EndFor
		\State \Return $\argmin_{\kappa \in \mcD }|\kappa|$.
	\end{algorithmic}
\end{algorithm}

\begin{lemma}\label{lem:MDSApx}
	Given graph $G$ and a $(1.5 - \eps)$-approximation algorithm for \gump, $\mcA$, \textsc{MDSApx} outputs in $\poly\left(n, \frac{1}{\eps} \right)$ time a dominating set of $G$ of size $O\left(\frac{k^*}{\eps} \right)$, where $k^*$ is the size of the minimum dominating set of $G$.
\end{lemma}
\begin{proof}
	By \Cref{lem:DSFromSched} we know that any set $\kappa_{\hat{k}} \in \mcD$ is a dominating set of $G$ of size at most $\frac{\Delta}{\eps}\left(|S_{\hat{k}}| - 2t_m - \Delta\right)$. Thus, it suffices to show that $\mcD$ contains at least one dominating set of $G$, $\kappa_{\hat{k}}$ such that $S_{\kappa_{\hat{k}}}$ is small. We do so now.
	
	Let $k^*$ be the size of the minimum dominating set of $G$. We know that $k^* \leq n$ and so in some iteration of $\textsc{MDSApx}$ we will have $\hat{k} = k^*$. Moreover, the minimum dominating set of $G_\alpha$ in this iteration just is $\frac{\Delta k^*}{\eps}$ since $G_\alpha$ is just $\frac{\Delta}{\eps}$ copies of $G$. Consider this iteration. Let $S^*$ be the optimal schedule for $G_\alpha'$ when $\hat{k} = k^*$. By \Cref{lem:optGumUpperBound} we know that $|S^*| \leq 2t_m + \Delta + \frac{k^*\Delta}{\eps}$. We now leverage the fact that that we chose $t_m$ to be \emph{large enough} so that $|S^*| < 3 t_m$. In particular, combining the fact that $|S^*| \leq 2t_m + \Delta + \frac{k^*\Delta}{\eps}$ with the fact that $\mcA$ is a $(1.5-\eps)$ approximation we have that 
	\begin{align}\label{eq:strictThreeTm}
	|S_{k^*}| &\leq (1.5 - \eps) |S^*|\nonumber\\
	&\leq (1.5 - \eps) \left(2t_m + \Delta + \frac{k^*\Delta}{\eps} \right)\nonumber\\
	&= 3t_m - 2\eps t_m + (1.5- \eps)\left(\Delta + \frac{k^*\Delta}{\eps}\right)\nonumber\\
	&= 3t_m - 2\eps t_m + (1.5- \eps)\eps\left(t_m - 1\right) \bc{$t_m$ dfn.}\nonumber\\
	&= 3t_m - \eps(0.5 + \eps) t_m - \eps(1.5 - \eps)\nonumber\\
	& < 3t_m.
	\end{align}
	Thus, since $|S_{k^*}| < 3 t_m$ we know that $\kappa_{k^*} \in \mcD$. Lastly, we argue that $|\kappa_{k^*}| = O\left(\frac{k^*}{\eps}\right)$, thereby showing that $\argmin_{\kappa \in \mcD }|\kappa|$, the returned dominating set of our algorithm, is $O\left(\frac{k^*}{\eps} \right)$.
	
	We now leverage the fact that we chose $t_m$ to be \emph{small enough} to give us a small dominating set. Applying \Cref{lem:DSFromSched} we have that
	\begin{align*}
	|\kappa_{k^*}| &\leq  \frac{\eps}{\Delta}\left(|S_{k^*}| - 2t_m - \Delta\right) \bc{\Cref{lem:DSFromSched}}\\
	&< \frac{\eps}{\Delta}(t_m - \Delta) \bc{\Cref{eq:strictThreeTm}}\\
	&= \frac{\eps}{\Delta}\left(\frac{1}{\eps}\left(\Delta + \frac{k^*\Delta}{\eps} \right) + 1 - \Delta \right) \bc{$t_m$ dfn.}\\
	&= \left(1 + \frac{k^*}{\eps} \right) + \frac{\eps}{\Delta} - \eps \\
	&= O\left(\frac{k^*}{\eps} \right)
	\end{align*}
	Thus, we conclude that \textsc{MDSApx} produces an $O\left(\frac{k^*}{\eps} \right)$ minimum dominating set of $G$.
	
	Lastly, we argue a polynomial in $n$ and $1/\epsilon$ runtime of \textsc{MDSApx}. First we argue that each iteration requires polynomial time. Constructing $G_\alpha$ takes polynomial time since the algorithm need only create $\frac{\Delta}{\eps} = \poly\left(n, \frac{1}{\eps} \right)$ copies of $G$. Running $\Psi$ also requires polynomial time since it simply adds polynomially many nodes to $G_\alpha$. $\mcA$ is polynomial by assumption and \textsc{DSFromSchedule} is polynomial by \Cref{lem:DSFromSched}. Thus, each iteration takes polynomial time and since \textsc{MDSApx} has $n$ iterations, \textsc{MDSApx} takes polynomial time in $n$ and $1/ \epsilon$.
\end{proof}

Given that \textsc{MDSApx} demonstrates an efficient approximation for minimum dominating set given a polynomial-time $(1.5 - \eps)$ approximation for \gump, we conclude our hardness of approximation.
\hardapx*
\begin{proof}
	Assume for the sake of contradiction that $\text{P} \neq \text{NP}$ and there existed a polynomial-time algorithm $\mcA$ that approximated \gump within $(1.5 - \eps)$ for $\eps = \frac{1}{o(\log n)}$. It follows by \Cref{lem:MDSApx} that \textsc{MDSApx} when run with $\mcA$ is a $o(\log n)$-approximation for minimum dominating set. However, this contradicts \Cref{lem:MDSHardnessAPX}, and so we conclude that \gump cannot be approximated within $(1.5 - \eps)$ for $\eps \geq \frac{1}{o(\log n)}$.
\end{proof}

\section{Omitted Lemmas of the Proof of Theorem~\ref{thm:apxalgo}}
\label{app:approx}

\subsection{Proof of Lemma \ref{lem:getPaths}}
\label{appsubsec:apxpaths}

The goal of this section is to prove Lemma~\ref{lem:getPaths}, which states the properties of \textsc{GetDirectedPaths}. To this end we will begin by rigorously defining the LP we use for \textsc{GetDirectedPaths} and establishing its relevant properties. We then formally define \textsc{GetDirectedPaths}, establish the properties of its subroutines and then prove \Cref{lem:getPaths}.

%\textsc{GetDirectedPaths} solves an LP which routes tokens from sources to sinks but might route tokens fractionally. The objective function of this LP is to minimize vertex congestion. To produce an integral solution from this LP, \textsc{GetDirectedPaths} performs a random walk for each token which takes edges proportionally to the fractional values output by the LP. If the resulting paths have high vertex congestion we resample random walks. \textsc{GetDirectedPaths} then carefully directs paths from the random walks.
\subsubsection{Our Flow LP}
The flow LP we use for \textsc{GetDirectedPaths} can be thought of as flow on a graph $G$ ``time-expanded'' by the maximum length that a \wad in the optimal schedule travels. Given any schedule we define the \emph{distance} that singleton \wad $a$ travels as the number of times any \wad containing $a$ is sent in said schedule. Let $L^*$ be the furthest distance a singleton \wad travels in the optimal schedule. Given a guess for $L^*$, namely $\hat{L}$, we define a graph $G_{\hat{L}}$ with vertices $\{v_r : v \in V, r \in [\hat{L}]\}$ and edges $\{e = (u_r, v_{r+1}) : (u, v) \in E, r \in [\hat{L} -1]\}$. We have a flow type for each $w \in W$, where $W = \{v : \text{$v$ has at least 1 \wad}\}$, which uses $\{w' : w' \in W \wedge w' \neq w\}$ as sinks. Correspondingly, we have a flow variable, $f_{w}(x_r,y_{r+1})$ for every $r \in [\hat{L} - 1]$, $w \in W$ and $(x,y) \in E$. The objective function of the LP is to minimize the maximum vertex congestion, given by variable $z$. Let $z(\hat{L})$ be the objective value of our LP given our guess $\hat{L}$. Formally, our LP is given in PathsFlowLP($\hat{L}$), where $\Gamma(v)$ gives the neighbors of $v$ in $G$. See \Cref{fig:LPIllustration} for an illustration of a feasible solution to this LP.

\vspace{15pt}

\begin{mdframed}[roundcorner=4pt, backgroundcolor=gray!5]
	\begin{small}
		\begin{align}\label{LP:getPaths}
		&\min z \text{ s.t.}\tag{PathsFlowLP($\hat{L}$)}\\
		&\text{\underline{``Conserve flow across rounds''}}\nonumber\\
		&\sum_{x' \in \Gamma(x)}f_w(x'_{r-1}, x_r) = \sum_{x'' \in \Gamma(x)}f_w(x_r, x''_{r+1})&\forall w \in W, x \not \in W, r\in[\hat{L} - 1]\label{line:consFlow}\\
		&\underline{\text{``Every $w \in W$ is a source for $f_w$ and not a sink for $f_w$''}}\nonumber&&\\
		& \sum_{r \in [\hat{L} - 1]}\left[\sum_{x' \in \Gamma(w)}f_w(w_r, x'_{r+1}) - \sum_{x' \in \Gamma(w)}f_w(x'_{r}, w_{r+1}) \right] \geq 1 &\forall w \in W \label{line:sourceSink}\\
		&\underline{\text{``$w$-flow ends at $w' \in W \text{ s.t. } w' \neq w$''}}\nonumber\\
		& \sum_{w' \in W : w' \neq w}\sum_{u \in \Gamma(w')}f_w(u_{\hat{L}},w'_{\hat{L}}) = 1 &\forall w\label{line:endAtS}\\
		&\underline{\text{``$z$ is the vertex congestion''}}\nonumber\\
		& z \geq \sum_w \sum_{v \in \Gamma(v)} \sum_{r \in [D - 1]} f_w(v'_r, v_{r+1}) &\forall v \label{line:vertCong}\\
		&\underline{\text{``Non-negative flow''}}\nonumber\\
		&f_w(x_r, y_{r+1}) \geq 0 &\forall , x, y, r, w \in W
		\end{align}
	\end{small}
\end{mdframed}

\bigskip

\begin{figure}
	\centering
	\includegraphics[width=\textwidth]{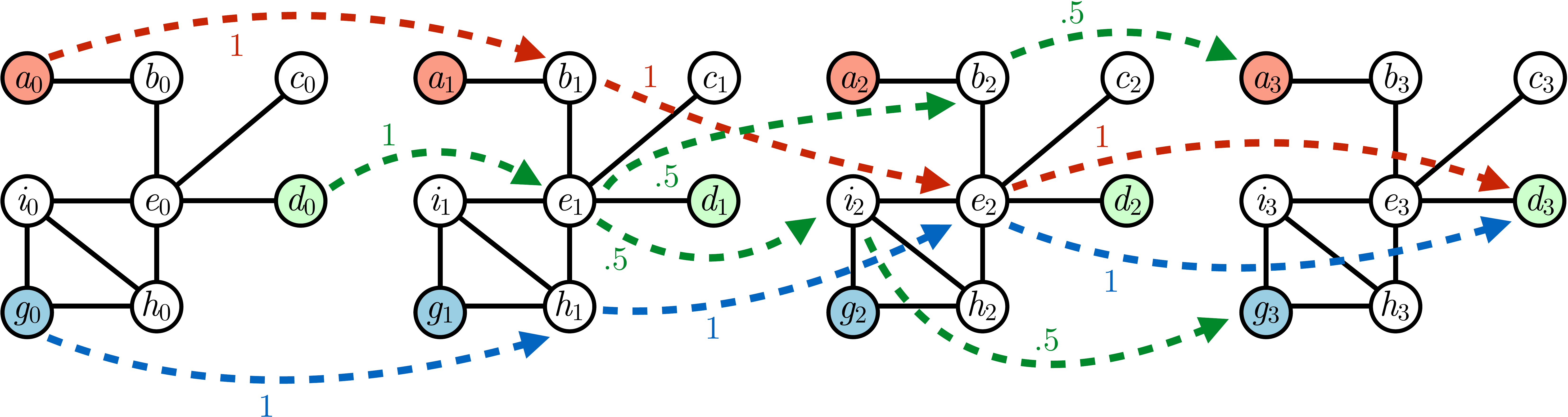}
	\caption{An illustration of non-zero flows for a feasible solution for PathsFlowLP(3) for graph $G$. Nodes $a$, $d$, and $g$ are in $W$, and $f_w$ is colored by $w$. For this feasible solution, $z = 2$.}
	\label{fig:LPIllustration}
\end{figure}

%As alluded to in the body of the paper, an important intermediate step is proving Lemma~\ref{lem:optGivesFeas}, which shows that we can produce a feasible solution for our LP of cost commensurate with $\OPT$. We first present the proof of Lemma~\ref{lem:optGivesFeas}. We then introduce the subroutines we use to sample paths, \textsc{SampleLPPaths}, and direct paths, \textsc{AssignPaths}, and prove their properties. Finally, we conclude with the proof of Lemma~\ref{lem:getPaths}.

\subsubsection{Proof of the Key Property of our LP}
The key property of our LP is that it has an optimal vertex congestion comparable to $\OPT$. In particular, we can produce a feasible solution for our LP of cost $2\OPT$ by routing tokens along the paths taken in the optimal schedule. 
%More formally, given any schedule we define the \emph{distance} that singleton \wad $a$ travels as the number of times any \wad containing $a$ is sent in said schedule. Let $L^*$ be the furthest distance a singleton \wad travels in the optimal schedule. Our LP is parameterized by a guess $\hat{L}$ for $L^*$. We let $z(\hat{L})$ be the value of our LP for a guess $\hat{L}$ of $L^*$. The following lemma demonstrates the mentioned relation between $\OPT$ and the optimal value of our LP.

\begin{restatable}{lemma}{optgivesfeas}\label{lem:optGivesFeas}
	$\min(t_c, t_m) \cdot z(2L^*) \leq 2  \OPT$.
\end{restatable}
The remainder of this section is a proof of \Cref{lem:optGivesFeas}.
Consider a $W$ as in \Cref{sec:prodPaths} where $W \gets \{v : \text{$v$ has at least 1 \wad}\}$ and the optimal schedule that solves \gump in time $\OPT$. %The key property of our LP that we use is that it has value commensurate with $\OPT$. Lemma \ref{lem:optGivesFeas} states this property. 

We will prove \Cref{lem:optGivesFeas} by showing that, by sending flow along paths taken by certain \wads in the optimal schedule, we can provide a feasible solution to \ref{LP:getPaths} with value commensurate with $\OPT$. For this reason we now formally define these paths, $\textsc{OptPaths}(W)$. Roughly, these are the paths taken by \wads containing singleton \wads that originate in $W$. Formally, these paths are as follows. Recall that $a_w$ is the singleton \wad with which node $w$ starts in the optimal schedule. Notice that in any given round of the optimal schedule exactly one \wad contains $a_w$. As such, order every round in which a \wad containing $a_w$ is received by a node in ascending order as $r_0(w), r_1(w) \ldots$ where we think of $w$ as receiving $a_w$ in the first round. Correspondingly, let $v_i(w)$ be the vertex that receives a \wad containing $a_w$ in round $r_i(w)$; that is $(v_1(w), v_2(w), \ldots)$ is the path ``traced out'' by $a_w$ in the optimal schedule. For \wad $a$, let $C(a) \coloneqq \{ a_{w'} : w' \in W \wedge a_w' \in a \}$ stand for all singleton \wads contained by \wad $a$ that originated at a $w' \in W$. Say that \wad $a$ is \emph{active} if $|C(a)|$ is odd. Let $v_{L_w}(w)$ be the first vertex in $(v_1(w), v_2(w), \ldots)$ where an active \wad containing $a_w$ is combined with another active \wad. Correspondingly, let $c(w)$ be the first round in which an active \wad containing $a_w$ is combined with another active \wad. Say that a singleton \wad $a_w$ is \emph{pending} in round $r$ if $r < c(w)$. We note the following behavior of pending singleton \wads.

\begin{lemma}\label{lem:exOnePend}
	In every round of the optimal schedule, if a \wad is active then it contains exactly one pending singleton \wad and if a \wad is inactive then it contains no pending singleton \wads.
\end{lemma}
\begin{proof}
	We prove this by induction over the rounds of the optimal schedule. As a base case, we note that in the first round of the optimal schedule a \wad is active iff it is a singleton node and every singleton node is pending. Now consider an arbitrary round $i$ and assume that our claim holds in previous rounds. Consider an arbitrary \wad $a$. If $a$ is not computed on by a node in this round then by our inductive hypothesis we have that it contains exactly one pending singleton \wad if it is active and no pending singleton \wads if it is not active. If $a$ is active and combined with an inactive \wad, by our inductive hypothesis, the resulting \wad contains exactly one pending singleton \wad. Lastly, if $a$ is active and combined with another active \wad by our inductive hypothesis these contain pending singletons $a_w$ and $a_u$ respectively such that $c(w) = c(u) = i$; it follows that the resulting \wad is inactive and contains no pending singleton \wads. This completes our induction.
\end{proof}

This behavior allows us to pair off vertices in $W$ based on how their singleton \wads are combined.\footnote{Without loss of generality we assume that $|W|$ is even here; if not, we can simply drop one element from $W$ each time we construct $\textsc{OptPaths}(W)$.}
\begin{lemma}\label{lem:optPathObs}
	For each $w \in W$ there exists a unique $u \in W$ such that $u \neq w$ and $v_{L_w}(w) = v_{L_u}(u)$ and $c(w) = c(u)$.
\end{lemma}
\begin{proof}
	Consider the round in which a \wad containing $a_w$, say $a$, is combined with an active \wad, say $b$, at vertex $v_{L_w}(w)$. Recall that this round is notated $c(w)$. By \Cref{lem:exOnePend} we know that $a$ and $b$ contain exactly one pending singleton \wad, say $a_w$ and $a_u$ respectively. Since both $a$ and $b$ are active in this round and $b$ contains $a_u$ we have $c(u) = c(w)$. Moreover, since both $a$ and $b$ are combined at the same vertex we have $v_{L_u}(u) = v_{L_w}(w)$. Lastly, notice that this $u$ is unique since by \Cref{lem:exOnePend} there is exactly one singleton \wad, $a_u$, contained by $b$ such that $c(u) \leq c(w)$.
\end{proof}

Having paired off vertices in $W$, we can now define $\textsc{OptPaths}(W)$. Fix a $w$ and let $u$ be the vertex it is paired off with as in \Cref{lem:optPathObs}. We define $\textsc{OptPath}(w) \coloneqq (v_1(w), v_2(w), \ldots v_{L_w}(w) = v_{L_u}(u), v_{L_u - 1}(u), \ldots, v_{1}(u)) $. Lastly, define $\textsc{OptPaths}(W) = \bigcup_{w \in W} \textsc{OptPath}(w)$. See \Cref{fig:optPaths} for an illustration of how $\textsc{OptPaths}(W)$ is constructed from the optimal schedule. 

\begin{figure}
	\centering
	\begin{subfigure}[t]{0.23\textwidth}
		\centering
		\includegraphics[scale=0.07,page=1]{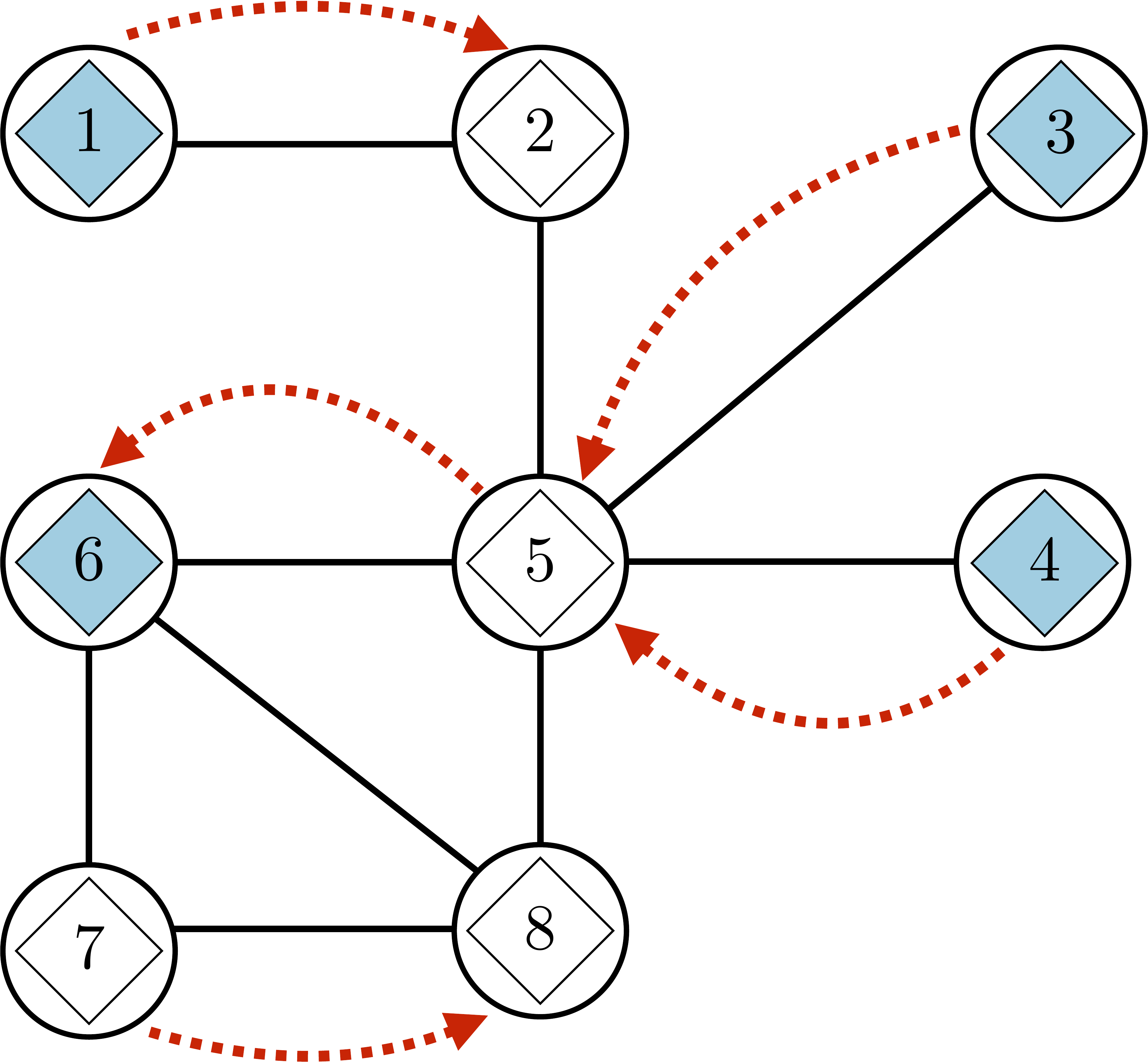}
		\caption{Round $1$}
	\end{subfigure}%
	~
	\begin{subfigure}[t]{0.23\textwidth}
		\centering
		\includegraphics[scale=0.07,page=2]{figs/optPaths.pdf}
		\caption{Round $2$}
	\end{subfigure}%
	~ 
	\begin{subfigure}[t]{0.23\textwidth}
		\centering
		\includegraphics[scale=0.07,page=3]{figs/optPaths.pdf}
		\caption{Round $3$}
	\end{subfigure}%
	~ 
	\begin{subfigure}[t]{0.23\textwidth}
		\centering
		\includegraphics[scale=0.07,page=4]{figs/optPaths.pdf}
		\caption{Round $4$}
	\end{subfigure}%
	
	\begin{subfigure}[t]{0.32\textwidth}
		\centering
		\includegraphics[scale=0.07,page=5]{figs/optPaths.pdf}
		\caption{Round $5$}
	\end{subfigure}%
	~
	\begin{subfigure}[t]{0.32\textwidth}
		\centering
		\includegraphics[scale=0.07,page=6]{figs/optPaths.pdf}
		\caption{Round $6$}
	\end{subfigure}%
	~ 
	\begin{subfigure}[t]{0.32\textwidth}
		\centering
		\includegraphics[scale=0.07,page=7]{figs/optPaths.pdf}
		\caption{Round $7$}
	\end{subfigure}%
	\caption{An illustration of the optimal schedule and how $\textsc{OptPaths}(W)$ is constructed from it for a particular $G$. Active \wads are denoted by blue diamonds; inactive \wads are denoted by white diamonds; a dotted red arrow from node $u$ to node $v$ means that $u$ sends to $v$; a double-ended blue arrow between two \wads $a$ and $b$ means that $a$ and $b$ are combined at the node; thick, dashed green lines give a path and its reversal in $\textsc{OptPaths}(W)$ (for a total of 4 paths across all rounds) where $(v_1(w), v_2(w), \ldots v_{L_w}(w) = v_{L_u}(u), v_{L_u - 1}(u), \ldots, v_{1}(u)) = P \in \textsc{OptPaths(w)}$ drawn only in round $c(w)$. Furthermore, \wad $a$ labeled with $\{v : \text{$a$ contains $a_v$}\}$ and $W = \{1, 3, 4, 6\}$. }\label{fig:optPaths}
\end{figure}

The critical property of $\textsc{OptPaths}(W)$ is that it has vertex congestion commensurate with $\OPT$ as follows.

\begin{lemma}\label{lem:conOP}
	$\con(\textsc{OptPaths}(W)) \leq \frac{2 \cdot \OPT}{ \min(t_c, t_m)}$.
\end{lemma}
\begin{proof}
	Call a pair of directed paths in $\textsc{OptPaths}(W)$ \emph{complementary} if one path is $\textsc{OptPath}(w)$ and the other $\textsc{OptPath}(u)$ where $u$ is to $w$ as in \Cref{lem:optPathObs}. We argue that each pair of complementary paths passing through a given vertex $v$ uniquely account for either $t_c$ or $t_m$ rounds of $v$'s $\OPT$ rounds in the optimal schedule. Consider a pair of complementary paths, $P = (\textsc{OptPath}(w), \textsc{OptPath}(u))$, passing through a given vertex $v$. This pair of paths pass through $v$ because in some round, say $r_P$, $v$ sends a \wad containing $a_u$ or $a_w$ or $v$ combines together \wads $a$ and $a'$ containing $a_u$ and $a_w$ respectively. Say that whichever of these operations accounts for $P$ is \emph{responsible} for $P$. Now suppose for the sake of contradiction that this operation of $v$ in round $r_P$ is responsible for another distinct pair $P'$ of complementary paths, $\textsc{OptPath}(w')$ and $\textsc{OptPath}(u')$. Notice that $a_w$, $a_{w'}$, $a_u$ and $a_{u'}$ are all pending in round $r_P$. We case on whether $v$'s action is a communication or a computation and show that $v$'s operation cannot be responsible for $P'$ in either case.
	\begin{itemize}
		\item Suppose that $v$ is responsible for $P$ and $P'$ because it performs a computation in $r_P$. It follows that $v$ combines an active \wad $a$ and another active \wad $a'$ where without loss of generality $a_w, a_{w'} \in a$ and $a_{u'}, a_{u} \in a'$. However, it then follows that $a$ is active and contains two pending singleton \wads, which contradicts \Cref{lem:exOnePend}.
		\item Suppose that $v$ is responsible for $P$ and $P'$ because it performs a communication in $r_P$ by sending \wad $a$. It follows that without loss of generality $a_w, a_{w'} \in a$. However, either $a$ is active or it is not. But by \Cref{lem:exOnePend} if $a$ is active it contains 1 pending singleton \wad and if $a$ is not active then it contains $0$ pending singleton \wads. Thus, the fact that $v$ sends a \wad containing two pending singleton \wads contradicts \Cref{lem:exOnePend}.
	\end{itemize}
	\noindent Thus, it must be the case that $v$'s action in $r_P$ is uniquely responsible for $P$.
	
	It follows that each computation and communication performed by $v$ uniquely corresponds to a pair of complementary paths (consisting of a pair of paths in $\textsc{OptPaths}(W)$) that passes through $v$. Since $v$ performs at most $\OPT / \min(t_c, t_m)$ operations in the optimal schedule, it follows that there are at most $\OPT / \min(t_c, t_m)$ pairs of complementary paths in $\textsc{OptPaths}(W)$ incident to $v$. Since each pair consists of two paths, there are at most $2 \cdot \OPT / \min(t_c, t_m)$ paths in $\textsc{OptPaths}(W)$ incident to $v$ and so $v$ has vertex congestion at most $2 \cdot \OPT / \min(t_c, t_m)$ in $\textsc{OptPaths}(W)$. Since $v$ was arbitrary, this bound on congestion holds for every vertex.
\end{proof}

We now use $\textsc{OptPaths}(W)$ to construct a  feasible solution for $\textsc{PathsFlowLP}(2L^*)$. We let $\tilde{f}$ be this feasible solution. Intuitively, $\tilde{f}$ simply sends flow along the paths of $\textsc{OptPaths}(W)$. More formally define $\tilde{f}$ as follows. For $w \in W$ and its corresponding path $\textsc{OptPath}(w) = (v_1(w), v_2(w), \ldots)$ we set $\tilde{f}_w(v_i, v_{i+1}) = 1$. We set all other variables of $\tilde{f}$ to 0 and let $\tilde{z}$ be the vertex congestion of $\textsc{OptPaths}(W)$.

\begin{lemma}\label{lem:feas}
	$(\tilde{f}, \tilde{z})$ is a feasible solution for $\textsc{PathsFlowLP}(2L^*)$ where \\ $\tilde{z} \leq 2 \OPT / \min(t_c, t_m)$.
\end{lemma}
\begin{proof}
	We begin by noting that every path in $\textsc{OptPaths}(W)$ is of length at most $2 L^*$: for each $w \in W$, $\textsc{OptPath}(w)$ is the concatenation of two paths, each of which is of length no more than $L^*$. Moreover, notice that for each $w \in W$, the sink of $\textsc{OptPath}(w)$ is a $w' \in W$ such that $w' \neq w$.
	
	We now argue that $(\tilde{f}, \tilde{z})$ is a feasible solution for $\textsc{PathsFlowLP}(2L^*)$: each vertex $v$ with incoming $w$-flow that is not in $W \setminus w$ sends out this unit of flow and so \Cref{line:consFlow} is satisfied; since each $\textsc{OptPath}(w)$ is of length at most $2 L^*$ and ends at a $w' \in W$ we have that every $w \in W$ is a source for $f_w$ and not a sink for $f_w$, satisfying  \Cref{line:sourceSink}; for the same reason, \Cref{line:endAtS} is satisfied; letting $\tilde{z}$ be the vertex congestion of $\textsc{OptPaths}(W)$ clearly satisfies \Cref{line:vertCong}; and flow is trivially non-zero.
	
	Lastly, since $\tilde{f}$ simply sends one unit of flow along each path in $\textsc{OptPaths}(W)$, our bound of $\tilde{z} \leq 2 \OPT / \min(t_c, t_m)$ follows immediately from \Cref{lem:conOP}.
\end{proof}

We conclude that $\tilde{f}$ demonstrates that our LP has value commensurate with $\OPT$.

%\begin{lemma}\label{lem:optGivesFeas}
%	$\min(t_c, t_m) \cdot z(2L^*) \leq 2  \OPT$.
%\end{lemma}
\optgivesfeas*
\begin{proof}
	Since $\Cref{lem:feas}$ shows that $(\tilde{f}, \tilde{z})$ is a feasible solution for $\textsc{PathsFlowLP}(2L^*)$ with cost at most $2 \OPT / \min(t_c, t_m)$, our claim immediately follows. 
\end{proof}

\subsubsection{\textsc{GetDirectedPaths} Formally Defined}
\textsc{GetDirectedPaths} solves our LP for different guesses of the longest path used by the optimal, samples paths based on the LP solution for our best guess, and then directs these paths. Formally, \textsc{GetDirectedPaths} is given in Algorithm~\ref{alg:getDirPaths}, where $\xi := \lceil 2(n-1) \cdot(t_c + D \cdot t_m) / t_m \rceil$ is the range over which we search for $L^*$.

\begin{algorithm} 
	\caption{\textsc{GetDirectedPaths}(G, W)}
	\label{alg:getDirPaths}
	\begin{algorithmic}
		\Statex \textbf{Input:} $W \subseteq V$ where $w \in W$ has a \wad
		\Statex \textbf{Output:} Directed paths between nodes in $W$
		\State $L \gets \argmin_{\hat{L} \in [\xi]} \left[ t_m \cdot \hat{L} + \min(t_c, t_m) \cdot t(\hat{L}) \right]$
		\State $f_w^* \gets \textsc{PathsFlowLP$(L)$}$
		\State $\mcP_W \gets \textsc{SampleLPPaths}(f_w^*, L, W)$
		\State $\vec{\mcP_U} \gets \textsc{AssignPaths}(\mcP_W, W)$
		\State \Return $\vec{\mcP_U}$
		
	\end{algorithmic}
\end{algorithm}

\subsubsection{Sampling Paths from LP}

Having shown that our LP has value commensurate with $\OPT$ and defined our algorithm based on this LP, we now provide the algorithm which we use to sample paths from our LP solution, \textsc{SampleLPPaths}. This algorithm produces a single sample by taking a random walk from each $w \in W$ where edges are taken with probability corresponding to their LP value. It repeats this $O(\log n)$ times to produce $O(\log n)$ samples. It then takes the sample with the most low congestion paths, discarding any high congestion paths in said sample. In particular, \textsc{SampleLPPaths} takes the sample $\mcP_W^i$ that maximizes $|Q(\mcP_W^i)|$ where $Q(\mcP_W^i) = \{P_w : P_w \in \mcP_W^i, \con(P_w) \leq 10 \cdot z(\hat{L}) \log \hat{L} \}$ for an input $\hat{L}$.

\begin{algorithm}
	\caption{\textsc{SampleLPPaths}($f^*_w$)}
	\label{alg:randRound}
	\begin{algorithmic}
		\Statex \textbf{Input:} $f^*_w$, solution to \ref{LP:getPaths}; $\hat{L}$, guess of $L^*$; $W \subseteq V$
		\Statex \textbf{Output:} Undirected paths between nodes in $W$
		\State $\mcC \gets \emptyset$
		\For{sample $i \in O(\log n)$}
		\State $\mcP_W^i \gets \emptyset$
		\For{$w \in W$}
		\State $v \sim f_w^*(w_1, v_2)$ 
		\State $P_w \gets (w, v)$
		\While {$v \not \in W$}
		\State $v' \sim f_w^*(v_{|P_w|}, v'_{|P_w| + 1})$
		\State $v \gets v'$
		\State $P_w += v$
		\EndWhile
		\State $\mcP_W^i \gets \mcP_W^i \cup \{ P_w \}$
		\EndFor	
		\State $\mcC \gets \mcC \cup \mcP_W^i$
		\EndFor
		
		\State $\mcP_S \gets Q(\argmax_{\mcP_W^i \in \mcC} |Q(\mcP_W^i)|)$
		
		\State \Return $\mcP_W$
	\end{algorithmic}
\end{algorithm}
\noindent The properties of \textsc{SampleLPPaths} are as follows.

\begin{lemma}\label{lem:randRound}
	For any fixed $W \subseteq V$, $\hat{L}$ and an optimal solution $f^*_w$ to \ref{LP:getPaths}, \textsc{SampleLPPaths} is a polynomial-time randomized algorithm that outputs a set of undirected paths $\mcP_W$ such that $P_w \in \mcP_W$ is an undirected path with endpoints $w, w' \in S$ where $w \neq w'$. Also $|\mcP_W| \geq \frac{1}{3}|W|$ w.h.p., $\con(\mcP_W) \leq z(\hat{L}) \cdot O(\log \hat{L})$, and $\dil(\mcP_W) \leq \hat{L}$.
\end{lemma}
\begin{proof}
	Our proof consists of a series of union and Chernoff bounds over our samples. Consider an arbitrary $W \subseteq V$. Define $T_w$ for $w\in W$ as the (directed) subgraph of $G_{\hat{L}}$ containing arc $(v, u) \in E_{\hat{L}}$ if for some $r$ we have $f^*_w(x,y) > 0$. Notice that $T_w$ is a weakly connected DAG where $w$ has no edges into it: $T_w$ does not contain any cycles since flow only moves from $x_r$ to $y_{r+1}$ for $x,y \in V$; by our flow constraints $T_w$ must be weakly connected and $w_r$ must have no edges into it for any $r$. Moreover, notice that $P_w$ is generated by a random walk on $T_w$ starting at $w_1$, where if the last vertex added to $P_w$ was $v$, then we add $u$ to $P_w$ in step $r$ of the random walk with probability $f^*_w(v_r,u_{r+1})$.
	
	We first argue that every $P_w \in \mcP_W$ has endpoints $w, w' \in W$ for $w \neq w'$ and $\dil(\mcP_W) \leq \hat{L}$.  By construction, one endpoint of $P_w$ is $w$. Moreover, the other endpoint of $P_w$ will necessarily be a $w' \in W$ such that $w' \neq w$: by \Cref{line:consFlow} flow is conserved and by \Cref{line:endAtS} all flow from $w$ must end at a point $w' \in W$ such that $w' \neq w$; thus our random walk will always eventually find such an $w'$. Moreover, notice that our random walk is of length at most $\hat{L} $ since $T_w$ is of depth at most $\hat{L}$. Thus, every $P_w$ is of length at most $\hat{L}$, meaning $\dil(\mcP_W) \leq \hat{L}$. 
	
	Next, notice that, by the definition of $Q$, $\con(\mcP_W) \leq z(\hat{L}) \cdot O(\log \hat{L})$ by construction since every element in $Q(\argmax_{\mcP_W^i \in \mcC} |Q(\mcP_W^i)|)$ has $O(z(\hat{L}) \cdot O(\log \hat{L}))$ congestion.
	
	Thus, it remains only to prove that $|\mcP_W| \geq \frac{1}{3}|W|$. We begin by arguing that for a fixed path $P_w$ in a fixed set of sampled paths, $\mcP_W^i$ we have $\con(P_w) \geq z(\hat{L}) \cdot O(\log \hat{L})$ with probability at most $\frac{1}{3}$. Consider a fixed path $P_w \in \mcP_W^i$ and fix an arbitrary $v \in P_w$. Now let $X_{wv}$ stand for the random variable indicating the number of times that path $P_w$ visits vertex $w$. without loss of generality we know that $P_w$ contains no cycles (since if it did we could just remove said cycles) and so $X_{sv}$ is either $1$ or $0$. %Notice that $X_{sv}$ is 1 if there exists a neighbor of $v$, $u$, and a round $r$ such that $u$ is taken in the $(r-1)$th step of our random walk and edge $(v,u)$ is taken in the $r$th round. 
	By a union bound over rounds, then, we have $\E[X_{wv}] \leq  \sum_r \sum_{u \in \Gamma(v)} f^*_W(u_r,v_{r+1}) \cdot \Pr(u \text{ taken in $(r-1)$th step}) \leq \sum_{u \in \Gamma(v)} \sum_r f^*_W(u_r,v_{r+1})$.
	
	Now note that the congestion of a single vertex under our solution is just $\con(v) = \sum_{w \in W} X_{wv}$. It follows that \[\E[\con(v)] = \sum_{w \in W} \E[X_{wv}] \leq \max_v \sum_w \sum_{u \in \Gamma(v)} \sum_r f^*_W(u_r,v_{r+1}) \leq z(\hat{L}).\] Also notice that for a fixed $v$ every $X_{wv}$ is independent. Thus, we have by a Chernoff bound that that
	
	\begin{align}\label{eq:vCong}
	\Pr(\con(v) \geq z(\hat{L}) \cdot O(\log \hat{L})) &\leq \Pr \left(\sum_{w \in W} X_{wv} \geq  \E\left[\sum_{w \in W} X_{wv} \right] \cdot O(\log \hat{L}) \right)\nonumber \\
	&\leq \frac{1}{(\hat{L})^c}
	\end{align}
	for $c$ given by constants of our choosing. $P_w$ is of length at most $\hat{L}$ by construction. Thus, by a union over $v \in P_w$ and \Cref{eq:vCong} we have that
	
	\begin{align*}
	\Pr\left(\con(P_w) \geq z(\hat{L}) \cdot O(\log \hat{L}) \right) &\leq \frac{1}{\hat{L}^{c-1}}\\
	&\leq \frac{1}{3}.
	\end{align*}
	\noindent Thus, for a fixed path $P_w \in \mcP_W^i$ we know that this path has congestion at least $z(\hat{L}) \cdot O(\log \hat{L}))$ with probability at most $\frac{1}{3}$.
	
	We now argue at least one of our $O(\log n)$ samples is such that at least $\frac{1}{3}$ of the paths in the sample have congestion at most $z(\hat{L}) \cdot O(\log \hat{L}))$. Let $Y_{iw}$ be the random variable that is $1$ if $P_w \in \mcP_W^i$ is such that $\con(P_w) \geq z(\hat{L}) \cdot O(\log \hat{L}))$ and 0 otherwise. Notice that $\E[Y_{iw}] \leq \frac{1}{3}$ by the fact that a path has congestion at least $z(\hat{L}) \cdot O(\log \hat{L})$ with probability at most $\frac{1}{3}$. Now let $Z_i = \sum_{w \in W}Y_{iw}$ stand for the number of paths in sample $i$ with high congestion. By linearity of expectation we have $\E[Z_i] \leq |W|\frac{1}{3}$. By Markov's inequality we have for a fixed $i$ that $\Pr(Z_i \geq \frac{2}{3}|W|) \leq \Pr(Z_i \geq 2\E[Z_i]|W|) \leq \frac{1}{2}$. Now consider the probability that every sample $i$ is such that more than $\frac{2}{3}$ of the paths have congestion more than $z(\hat{L}) \cdot O(\log \hat{L})$, i.e.\ consider the probability that for all $i$ we have $Z_i \geq |W|\frac{2}{3}$. We have
	
	\begin{align*}
	\Pr\left(Z_i \geq |W|\frac{2}{3}, \forall i \right) &\leq \left(\frac{1}{2} \right)^{O(\log n)}\\
	&= \frac{1}{ \poly(n)}.
	\end{align*}
	Thus, with high probability there will be some sample, $i$, such that $Z_i \leq |W|\frac{2}{3}$. It follows that with high probability $\max_{\mcP_W^i \in \mcC} |Q(\mcP_W^i)| \geq \frac{1}{3}|W|$ and since $\mcP_W = Q(\argmax_{\mcP_W^i \in \mcC} |Q(\mcP_W^i)|)$, we conclude that with high probability $\mcP_W \geq \frac{1}{3} |W|$.
	
\end{proof}

\subsubsection{Directing Paths}

Given the undirected paths that we sample from our LP, $\mcP_W$, we produce a set of directed paths $\vec{\mcP_U}$ using \textsc{AssignPaths}, which works as follows. Define $G'$ as the directed supergraph consisting of nodes $W$ and directed edges $E' = \{(w, w') : \text{$w'$ is an endpoint of $P_w \in \mcP_W$})\}$. Let $\Gamma_{G'}(v) = \{v' : (v', v) \in E' \lor (v , v') \in E' \}$ give the neighbors of $v$ in $G'$. For each node $w \in G'$ with in-degree of at least two we do the following: if $v$ has odd degree delete an arbitrary neighbor of $w$ from $G'$; arbitrarily pair off the neighbors of $w$; for each such pair $(w_1, w_2)$ add the directed path $P_{w_1} \circ rev(P_{w_2})$ to $\vec{\mcP_U}$ where $rev(P_{w_2})$ gives the result of removing the last element of $P_{w_2}$ (namely, $w$) and reversing the direction of the path; remove $\{w, w_1, w_2\}$ from $G'$. Since we remove all vertices with in-degree of two or more and every vertex has out-degree $1$, the remaining graph trivially consists only of nodes with in-degree at most 1 and out-degree at most 1. The remaining graph, therefore, is all cycles and paths. For each cycle or path $w_1, w_2, w_3, \ldots$ add the path corresponding to the edge from $w_i$ to $w_{i+1}$ for odd $i$ to $\vec{\mcP}_U$. We let $U$ be all sources of paths in $\vec{\mcP}_U$ and we let $P_u$ be the path in $\vec{\mcP}_U$ with source $u$.

%\begin{algorithm}
%	\caption{$\textsc{AssignPaths}(\mcP_W, W)$}
%	\label{alg:randRound}
%	\begin{algorithmic}
%		\Statex \textbf{Input:} undirected paths between nodes in $W$, $\mcP_W$; $W \subseteq V$
%		\Statex \textbf{Output:} directed paths between nodes in $W$, $\vec{\mcP_W}$
%		\State $G' \gets (W, \{(w, w') : \exists \text{path $P_{ww'} \in \mcP_W$  with endpoints $w$ and $w'$}\})$
%		\State $\vec{\mcP}_W \gets \emptyset$		
%		\For{$v \in G'$}
%		\State Let $\{u_1, u_2, \ldots, u_{|\Gamma_{G'}(v)|}\}$ be the nodes in $\Gamma_{G'}(v)$
%%			\If{$|\Gamma_{G'}(v)|$ is odd}
%%				\State $G \gets G \setminus u_{|\Gamma_{G'}(v)|}$
%%			\EndIf
%
%				\For{$u_i \in \Gamma_{G'}(v)$ where $i$ is even}
%					\State $\vec{\mcP}_W  \gets (\vec{P_{u_iv}} \circ \vec{P_{vu_{i-1}}}) \cup \vec{\mcP}_W $
%				\State $G' \gets G' \setminus \{u_i, u_{i-1}\}$
%
%				\EndFor
%				\State $G' \gets G' \setminus v$
%		\EndFor
%		\State \Return $\vec{\mcP_W}$
%	\end{algorithmic}
%\end{algorithm}

The properties of \textsc{AssignPaths} are as follows.

\begin{lemma}\label{lem:dirPaths}
	Given $W \subseteq V$ and $\mcP_W = \{P_w : w \in W\}$ where the endpoints of $P_w$ are $w, w' \in W$ for $w \neq w'$, \textsc{AssignPaths} in polynomial-time returns directed paths $\vec{\mcP_U}$ where at least $1/4$ of the nodes in $W$ are the source of a directed path in $\vec{\mcP_U}$, each path in $\vec{\mcP_U}$ is of length at most $2 \cdot \dil(\mcP_W)$ with congestion at most $\con(\mcP_W)$ and each path in $\vec{\mcP_U}$ ends in a unique sink in $W$.
\end{lemma}
\begin{proof}
	When we add paths to $\vec{\mcP}_U$ that go through vertices of in-degree at least two, for every 4 vertices we remove we add at least one directed path to $\vec{\mcP}_U$ that is at most double the length of the longest a path in $\mcP_U$: in the worst case $v$ has odd in-degree of 3 and we add only a single path. When we do the same for our cycles and paths for every 3 vertices we remove we add at least one directed path to $\vec{\mcP}_U$. Notice that by construction we clearly never reuse sinks in our directed paths. The bound on congestion and a polynomial runtime are trivial.
\end{proof}

%Lastly, we specify \textsc{GetDirectedPaths}. \textsc{GetDirectedPaths} will solve our LP for different guesses of the longest path used by the optimal, sample paths based on the LP solution for our best guess and then use \textsc{AssignPaths} to direct these paths. Formally, our algorithm is as follows where $\xi := \lceil 2(n-1) \cdot(t_c + D \cdot t_m) / t_m \rceil$ is the range over which we search for $L^*$.
%
%\begin{algorithm} 
%	\caption{\textsc{GetDirectedPaths}(G, W)}
%	\label{alg:getDirPaths}
%	\begin{algorithmic}
%		\Statex \textbf{Input:} $W \subseteq V$ where $w \in W$ has a \wad
%		\Statex \textbf{Output:} Directed paths between nodes in $W$
%		\State $L \gets \argmin_{\hat{L} \in [\xi]} \left[ t_m \cdot \hat{L} + \min(t_c, t_m) \cdot t(\hat{L}) \right]$
%		\State $f_w^* \gets \textsc{PathsFlowLP$(L)$}$
%		\State $\mcP_W \gets \textsc{SampleLPPaths}(f_w^*, L, W)$
%		\State $\vec{\mcP_U} \gets \textsc{AssignPaths}(\mcP_W, W)$
%		\State \Return $\vec{\mcP_U}$
%		
%	\end{algorithmic}
%\end{algorithm}

\subsubsection{Proof of Lemma~\ref{lem:getPaths}}
Finally, we conclude with the proof of Lemma~\ref{lem:getPaths}.

%\begin{lemma}\label{lem:getPaths}
%	Given $W \subseteq V$, \textsc{GetDirectedPaths} is a randomized polynomial-time algorithm that returns a set of directed paths, $\vec{\mcP_U} = \{P_u : u \in U\}$ for $U \subseteq W$, such that with high probability at least $1/12$ of nodes in $W$ are sources of paths in $\vec{\mcP_U}$ each with a unique sink in $W$. Moreover, $\con(\vec{\mcP_U}) \leq O\left(\frac {\OPT}{\min(t_c, t_m)} \log \frac{\OPT}{t_m} \right)$ and $\dil(\vec{\mcP_U}) \leq \frac{8 \OPT}{t_m}$.
%\end{lemma}
\getpaths*
\begin{proof}
	The fact that \textsc{GetDirectedPaths} returns a set of directed paths, $\vec{\mcP_U}$, such that at least $1/12$ of nodes in $W$ are sources in a path with a sink in $W$ follows directly from \Cref{lem:randRound} and \Cref{lem:dirPaths}. 
	
	We now give the stated bounds on congestion and dilation. First notice that $2L^* \in [\xi]$. Moreover, $2\OPT \leq 2(n-1)(t_c + D \cdot t_m)$: the schedule that picks a pair of nodes, routes one to the other then aggregates and repeats $n-1$ times is always feasible and takes $(n-1)(t_c + D \cdot t_m)$ rounds. Thus, $2L^* \leq 2\frac{\OPT}{t_m} \leq \xi$.
	
	Thus, by definition of $L$ we know that 
	\begin{align*}
	t_m \cdot L + \min(t_c, t_m) \cdot t(L) &\leq 2t_m \cdot L^*  + \min(t_c, t_m) \cdot z(2 L^*)\\
	& \leq 2L^* + 2 \OPT \bc{\Cref{lem:optGivesFeas}}\\
	& \leq 4 \OPT \bc{dfn.\ of $L^*$}
	\end{align*}
	
	It follows, then, that $t_m \cdot L \leq 4 \OPT$ and so $L \leq \frac{4 \OPT}{t_m}$. Similarly, we know that $\min(t_c, t_m) \cdot z(L) \leq 4 \OPT$ and so $z(L) \leq \frac{4 \OPT}{\min(t_c, t_m)}$. 
	
	Lastly, by \Cref{lem:randRound} we know that $\dil(\mcP_W) \leq L \leq \frac{4 \OPT}{t_m}$ and $\con(\mcP_W) \leq t(L) \cdot O(\log L) \leq O\left(\frac{ \OPT}{\min(t_c, t_m)} \cdot \log \frac{\OPT}{t_m} \right)$. By \Cref{lem:dirPaths} we get that the same congestion bound holds for $\vec{\mcP_U}$ and $\dil(\vec{\mcP_U}) \leq \frac{8 \OPT}{\min(t_c, t_m)}$.
	
	A polynomial runtime comes from the fact that we solve at most $(n-1)(t_c + D \cdot t_m) = \poly(n)$ LPs and then sample at most $(n-1)(t_c + D \cdot t_m)$ edges $O(\log n)$ times to round the chosen LP.
\end{proof}

\subsection{Deferred Proofs of Section \ref{sec:routingAlongProdPaths}}
\label{appsubsec:apxroute}

\optroute*
\begin{proof}
	
	Given a set of paths $\vec{\mcP_U}$, Rothvo\ss~\cite{rothvoss2013simpler} provides a polynomial-time algorithm that produces a schedule that routes along all paths in $O(\text{con}_E(\vec{\mcP_U}) + \dil(\vec{\mcP_U})$ where $\text{con}_E(\mcP) = \max_e \sum_{P \in \mcP} \mathbbm{1}(e \in P)$ is the edge congestion. However, the algorithm of Rothvo\ss~\cite{rothvoss2013simpler} assumes that in each round a vertex can send a \wad along each of its incident edges whereas we assume that in each round a vertex can only forward a single \wad. 
	
	However, it is easy to use the algorithm of Rothvo\ss~\cite{rothvoss2013simpler} to produce an algorithm that produces a \gum routing schedule using $O(\text{con}(\vec{\mcP_U}) + \dil(\vec{\mcP_U}))$ rounds which assumes that vertices only send one \wad per round as we assume in the \gum model as follows. Let $G$ be our input network with paths $\vec{\mcP_U}$ along which we would like to route where we assume that vertices can only send one \wad per round. We will produce another graph $G'$ on which to run the algorithm of Rothvo\ss~\cite{rothvoss2013simpler}. For each node $v \in G$ add nodes $v_i$ and $v_o$ to $G'$. Project each path $P \in \vec{\mcP_U}$ into $G'$ to get $P' \in \vec{\mcP_U}'$ as follows: if edge $(u,v)$ is in path $P \in \vec{\mcP_S}$ then add edge $(u_o, v_i)$ and edge $(v_i, v_o)$ to path $P'$ in $G'$. Notice that $\text{con}(\vec{\mcP_U}) = \text{con}_E(\vec{\mcP_U}')$ and $\dil(\vec{\mcP_U}) = 2\dil(\vec{\mcP_U}')$. Now run the algorithm of Rothvo\ss~\cite{rothvoss2013simpler} on $G'$ with paths $\mcP'_U$ to get back some routing schedule $S'$.
	
	Without loss of generality we can assume that $S'$ only has nodes in $G'$ send along a single edge in each round: every $v_i$ is incident to a single outbound edge across all paths (namely $(v_i, v_o)$) and so cannot send more than one \wad per round; every $v_o$ has a single incoming edge and so receives at most one \wad per round which, without loss of generality, we can assume $v_o$ sends as soon as it receives (it might be the case that $v_o$ collects some number of \wads over several rounds and then sends them all out at once but we can always just have $v_o$ forward these \wads as soon as they are received and have the recipients ``pretend'' that they do not receive them until $v_o$ would have sent out many \wads at once). 
	
	Now generate a routing schedule for $G$ as follows: if $v_o$ sends \wad $a$ in round $r$ of $S'$ then $v$ will send \wad $a$ in round $r$ of $S$. Since $S$ only ever has vertices send one \wad per round, it is easy to see by induction over rounds that $S$ will successfully route along all paths. Moreover, $S$ takes as many rounds as $S'$ which by \cite{rothvoss2013simpler} we know takes $O(\con(\vec{\mcP'_U}) + \dil(\vec{\mcP'_U})) = O(\con(\vec{\mcP_U}) + 2 \dil(\vec{\mcP_U})) = O(\con(\vec{\mcP_U}) + \dil(\vec{\mcP_U}))$. Thus, we let $\textsc{OPTRoute}$ be the algorithm that returns $S$.
\end{proof}

\routepathsm*
\begin{proof}
	By \Cref{lem:OPTRoute}, \textsc{OPTRoute} takes $t_m (\con(\vec{\mcP_U}) +  \dil(\vec{\mcP_U}))$ rounds to route all sources to sinks. All sources are combined with sinks in the following computation and so $\textsc{RoutePaths}_m$ successfully solves the \textsc{Route and Compute} Problem since every source has its \wad combined with another \wad. The polynomial runtime of the algorithm is trivial.
\end{proof}

\routepathsc*
\begin{proof}
	We argue that every source's \wad ends at an asleep node with at least two \wads and no more than $\con(\vec{\mcP_U})$ \wads. It follows that our computation at the end at least halves the number of \wads.
	
	First notice that if a vertex falls asleep then it will receive at most $\con(\vec{\mcP_S})$ \wads by the end of our algorithm since it is incident to at most this many paths. Moreover, notice that every \wad will either end at a sink or a sleeping vertex and every sleeping vertex is asleep because it has two or more \wads. It follows that every \wad is combined with at least one other \wad and so our schedule at least halves the total number of \wads.
	
	The length of our schedule simply comes from noting that we have $O(\dil(\vec{\mcP_U}) \cdot t_m)$ forwarding rounds followed by $\con(\vec{\mcP_U}) \cdot t_c$ rounds of computation. Thus, we get a schedule of total length $O(t_c \cdot \con(\vec{\mcP_S}) + t_m \cdot \dil(\vec{\mcP_S}))$. A polynomial runtime is trivial.
\end{proof}

\subsection{Proof of \Cref{thm:apxalgo}}
\label{appsubsec:apxalgoproof}
\apxthm*
\begin{proof}
	By \Cref{lem:getPaths} we know that the paths returned by \textsc{GetDirectedPaths}, $\vec{\mcP_U}$ are such that $\con(\vec{\mcP_U}) \leq O\left(\frac {\OPT}{\min(t_c, t_m)} \log \frac{\OPT}{t_m} \right)$ and $\dil(\vec{\mcP_U}) \leq \frac{8 \OPT}{t_m}$ and the paths returned have unique sinks and sources in $W$ and there are at least $|W|/12$ paths w.h.p. 
	
	If $t_c > t_m$ then $\textsc{RoutePaths}_m$ is run which by \Cref{lem:routePathsM} solves the \textsc{Route and Compute} Problem in  $O(t_m \cdot \con(\vec{\mcP_U})  + t_m \cdot \dil(\vec{\mcP_U}) + t_c)$ rounds which is
	\begin{align*}
	&\leq O\left(t_m \cdot \frac { \OPT}{\min(t_c, t_m)} \cdot \log \frac{\OPT}{t_m}  + t_m \cdot \frac{8 \OPT}{t_m}  + t_c \right)\\
	&= O\left(\OPT \cdot \log \frac{\OPT}{t_m} + t_c \right)
	\end{align*}
	
	If $t_c \leq t_m$ then $\textsc{RoutePaths}_c$ is run to solve the \textsc{Route and Compute} Problem which by \Cref{lem:routePathsC} takes $O(t_c \cdot \con(\vec{\mcP_U})  + t_m \cdot \dil(\vec{\mcP_U}))$ rounds which is
	
	\begin{align*}
	&\leq O\left(t_c \cdot \frac {4 \OPT}{\min(t_c, t_m)} \cdot \log \frac{\OPT}{t_m}  + t_m \cdot \frac{8 \OPT}{t_m} \right)\\
	&= O\left(\OPT \cdot \log \frac{\OPT}{t_m} \right)
	\end{align*}
	
	Thus, in either case, the produced schedule takes at most $O \left(\OPT \cdot \log \frac{\OPT}{t_m} + t_c \right)$ rounds to solve the \textsc{Route and Compute} Problem on at least $|W|/12$ paths in each iteration. 
	Since solving the \textsc{Route and Compute} Problem reduces the total number of \wads by a constant fraction on the paths over which it is solved, and we have at least $|W|/12$ paths in each iteration w.h.p., by a union bound, every iteration reduces the total number of \wads by a constant fraction w.h.p.
	Thus, the concatenation of the $O(\log n)$ schedules produced, each of length $O(\OPT \cdot \log \frac{\OPT}{t_m} + t_c)$, is sufficient to reduce the total number of \wads to $1$. 
	
	Thus, \solvegump produces a schedule that solves the problem of \gump in $O(\OPT \cdot \log n \log \frac{\OPT}{t_m} + t_c \cdot \log n)$ rounds. However, notice that $t_c \cdot \log n \leq \OPT$ (since the optimal schedule must perform at least $\log n$ serialized computations) and so the produced schedule is of length $O(\OPT \cdot \log n \log \frac{\OPT}{t_m} + t_c \log n) \leq O(\OPT \cdot \log n \log \frac{\OPT}{t_m})$. Lastly, a polynomial runtime is trivial given the polynomial runtime of our subroutines.
\end{proof}

\end{document}